\documentclass[11pt]{article}
\usepackage[utf8]{inputenc}
\usepackage{csquotes}
\usepackage{times}
\usepackage{tikz,comment,amsfonts,amssymb,amsmath,amsthm,bigints}
\usepackage[toc,page,header]{appendix}
\usepackage[T1]{fontenc}
\usepackage{hyperref}
\usepackage{fullpage,comment,algorithm,algorithmic}
\usepackage[colorinlistoftodos,textsize=tiny,textwidth=2cm,color=red!25!white,obeyFinal]{todonotes}
\numberwithin{equation}{subsection}
\usepackage{footnote}
\makesavenoteenv{tabular}
\usepackage{mathtools}
\usepackage{amsthm}

\DeclarePairedDelimiter\abs{\lvert}{\rvert}

\usepackage[colorinlistoftodos,textsize=tiny,textwidth=2cm,color=red!25!white,obeyFinal]{todonotes}

\newcommand\lqdist{\mathop{\mbox{$\ell_q$-$\mathit{dist}$}}}
\newcommand{\ldist}[1]{\mathop{\mbox{$\ell_{#1}$-$\mathit{dist}$}}}

\newcommand{\norm}[1]{\left\|#1\right\|}

\newcommand{\lexpans}[1]{\mathop{\mbox{$\ell_{#1}$-$\mathit{expans}$}}}

\newcommand{\expans}{\mathop{\mbox{$expans$}}}

\newcommand{\distort}{\mathop{\mbox{$dist$}}}

\newcommand\senergy{\mathop{\mbox{$REM$}}}
\newcommand\rem{\mathop{\mbox{$REM$}}}
\newcommand\energy{\mathop{\mbox{$Energy$}}}

\newcommand{\E}{{\rm E}}
\newcommand\inp[2]{\langle #1, #2 \rangle}

\newcommand{\eps}{{\epsilon}}

\usetikzlibrary{calc}

\newcommand{\alert}[1]{\textbf{\color{red}
[#1]}\marginpar{\textbf{\color{red}**}}\typeout{ALERT:
\the\inputlineno: #1}}

\newtheorem{theorem}{Theorem}
\newtheorem{claim}[theorem]{Claim}
\numberwithin{theorem}{section}
\newtheorem{definition}{Definition}
\numberwithin{definition}{section}

\newtheorem{lemma}[theorem]{Lemma}

\numberwithin{fact}{section}

\newtheorem{corollary}{Corollary}
\numberwithin{corollary}{section}

\title{Optimality of the Johnson-Lindenstrauss Dimensionality Reduction for Practical Measures}

\begin{document}

\author{
Yair Bartal
\thanks{Hebrew University.  Supported in part by a grant from the Israeli Science Foundation (1817/17). \href{mailto:yair@cs.huji.ac.il}{\texttt{yair@cs.huji.ac.il}}} \\
\and
Ora Nova Fandina 
\thanks{Aarhus University. \href{mailto:fandina@gmail.com}{\texttt{fandina@cs.au.dk}}} \\ 
 \\
\and
Kasper Green Larsen 
\thanks{Aarhus University. \href{mailto:larsen@cs.au.dk }{\texttt{larsen@cs.au.dk }}} 
}
\date{}

\begin{titlepage}
\def\thepage{}
\thispagestyle{empty}
\maketitle
\begin{abstract}
It is well known that the Johnson-Lindenstrauss dimensionality reduction method is optimal for worst case distortion. While in practice many other methods and heuristics are used, not much is known in terms of bounds on their performance. The question of whether the JL method is optimal for practical measures of distortion was recently raised in \cite{BFN19} (NeurIPS'19). They provided upper bounds on its quality for a wide range of practical measures and showed that indeed these are best possible in many cases. Yet, some of the most important cases, including the fundamental case of average distortion were left open. In particular, they show that the JL transform has $1+\epsilon$ average distortion for embedding into $k$-dimensional Euclidean space, where $k=O(1/\eps^2)$, and for more general $q$-norms of distortion, $k = O(\max\{1/\eps^2,q/\eps\})$, whereas tight lower bounds were established only for large values of $q$ via reduction to the worst case.

In this paper we prove that these bounds are best possible for any dimensionality reduction method, for any $1 \leq q \leq O(\frac{\log (2\eps^2 n)}{\eps})$ and $\epsilon \geq \frac{1}{\sqrt{n}}$, where $n$ is the size of the subset of Euclidean space.

Our results also imply that the JL method is optimal for various distortion measures commonly used in practice, such as {\it stress, energy} and {\it relative error}. We prove that if any of these measures is bounded by $\eps$ then $k=\Omega(1/\eps^2)$, for any $\epsilon \geq \frac{1}{\sqrt{n}}$, matching the upper bounds of \cite{BFN19} and extending their tightness results for the full range moment analysis.

Our results may indicate that the JL dimensionality reduction method should be considered more often in practical applications, and the bounds we provide for its quality should be served as a measure for comparison when evaluating the performance of other methods and heuristics.
\end{abstract}
\end{titlepage}

\section{Introduction}

Dimensionality reduction is a key tool in numerous fields of data analysis, commonly used as a compression scheme to enable reliable and efficient computation. In metric dimensionality reduction subsets of high-dimensional spaces are embedded into a low-dimensional space, attempting to preserve metric structure of the input. There is a large body of theoretical and applied research on such methods spanning a wide range of application areas such as online algorithms, computer vision, network design, machine learning, to name a few.

Metric embedding has been extensively studied by mathematicians and computer scientists over the past few decades (see \cite{Indyk01,linial,IndykMat04} for surveys). In addition to the beautiful theory, a plethora of original and elegant techniques have been developed and successfully applied in various fields of algorithmic research, e.g., clustering, nearest-neighbor, distance oracle. See \cite{Mat02, Indyk01, Vempala} for exposition of some applications.

The vast majority of these methods have been designed to optimize the worst-case distance error incurred by embedding.
For metric spaces $(X, d_X)$ and $(Y,d_Y)$ an injective map $f:X \to Y$ is an embedding. It has (a worst-case) distortion $\alpha \geq 1$ if there is a positive constant $c$ satisfying for all $u\neq v \in X$, $ d_Y(f(u), f(v))\leq c\cdot d_X(u,v) \leq \alpha \cdot d_Y(f(u), f(v))$. A cornerstone result in metric dimensionality reduction is the celebrated Johnson-Lindenstrauss lemma \cite{JL}. It states that any $n$-point subset of Euclidean space can be embedded, via a linear transform, into a $O(\log n/\epsilon^2)$-dimensional subspace with $1+\epsilon$ distortion. It has been recently shown to be optimal in $\cite{LN17}$ and in $\cite{AK17}$ (improving upon \cite{Alon09}). Furthermore, it was shown in \cite{Maush90} that there are Euclidean pointsets in $\mathbb{R}^d$ for which any embedding into $k$-dimensions must have $n^{\Omega(1/k)}$ distortion, effectively ruling out dimensionality reduction into a constant number of dimensions with a constant worst-case distortion.

Metric embedding and in particular dimensionality reduction have also gained significant attention in applied community.
Practitioners have frequently employed classic tools of metric embedding as well as have designed new techniques to cope with high-dimensional data. A large number of dimensionality reduction heuristics have been developed for a variety of practical settings, eg. \cite{tSNE, Umap, TriMap, PacMap}. However, most of these heuristics have not been rigorously analyzed in terms of the incurred error. Recent papers \cite{VL18} and \cite{BFN19} initiate the formal study of practically oriented analysis of metric embedding.

\paragraph*{Practical distortion measures}
In contrast to the worst case distortion the quality of practically motivated embedding is often determined by its average performance over all pairs, where an error per pair is measured as an additive error, a multiplicative error or a combination of both. There is a huge body of applied research investigating such notions of quality. For the list of citations and a more detailed account on the theoretical and practical importance of average distortion measures see \cite{BFN19}. 

In this paper we consider the most basic and commonly used in practical applications notions of average distortion, which we define in the following. 
Moment of distortion was defined in \cite{ABN06}, and studied in various papers since then.
\begin{definition}[$\ell_q$-distortion]\label{def:av_dist}

For $u \neq v \in X$ let $expans_f(u,v)=d_Y(f(u),f(v))/d_X(u,v)$ and \linebreak $contract_f(u,v)=d_X(u,v)/d_Y(f(u),f(v))$. 
Let $dist_f(u,v)=\max\{expans_f(u,v),contract_f(u,v)\}$. For any $q\geq 1$ the $q$-th moment of distortion is defined by 

\[\ldist{q}(f)=\left( \frac{1}{{  \binom{\abs{X}}{2}}}\sum_{u \neq v \in X}(dist_f(u,v))^q\right)^{1/q}.\] 

\end{definition}
Additive average distortion measures are often used when a nearly isometric embedding is expected. Such notions as {\it energy}, {\it stress} and {\it relative error measure} (REM) are common in various statistic related applications. For a map $f:X \to Y$, for a pair $u\neq v \in X$ let $d_{u,v}:=d_X(u,v)$ and let $\hat{d}_{uv}:=d_Y(f(u),f(v))$. For all $q\geq 1$

\begin{definition}[Additive measures] 
\[
{Energy}_{q}(f) =
\left(\frac{1}{{\binom{\abs{X}}{2}}}\sum_{u \neq v \in X} \left(\frac{\abs{\hat{d}_{uv}-d_{uv}}}{d_{u,v}} \right)^q\right)^{\frac{1}{q}} 
 =\left(\frac{1}{{\binom{\abs{X}}{2}}}\sum_{u \neq v \in X}\abs[\big]{expans_f(u,v)-1}^q\right)^{\frac{1}{q}} .\]

\[{Stress}_{q}(f) = \left(\frac{ \sum_{u \neq v \in X} |\hat d_{uv}-d_{uv}|^q }{ \sum_{u \neq v \in X}(d_{uv})^q}\right)^{\frac{1}{q}},\;\;\; {Stress^*}_{q}(f)=\left( \frac{\sum_{u \neq v \in X} |\hat d_{uv}-d_{uv}|^q  }{ \sum_{u\neq v \in X}(\hat d_{uv})^q}\right)^{\frac{1}{q}}.\]

\[{\senergy}_{q}(f)  =   {\left( \frac{1}{{\binom{\abs{X}}{2}}}\sum_{u \neq v \in X}\left( \frac{|\hat d_{uv} -d_{uv}|}{\min\{\hat d_{uv}, d_{uv}\}} \right)^q \right)}^{\frac{1}{q}}.
\]
\end{definition}

It was proved in \cite{BFN19} that 
\begin{claim}\label{BFN-relations}
 For all $q \geq 1$, 
$\lqdist(f) - 1 \geq {\rem}_q(f) \geq {\energy}_q(f)$. 
\end{claim}


Finally, \cite{CV18} defined $\sigma$-distortion and showed it to be particularly useful in machine learning applications. For $r\geq 1$, let $\lexpans{r}(f)=({\binom{n}{2}}^{-1}\sum_{u \neq v } ({\expans}_{f}(u,v))^r)^{1/r}$

\begin{definition}[$\sigma$-distortion]
\[\sigma_{q,r}(f)= \left( \frac{1}{ \binom{\abs{X}}{2}}\left|\frac{{\expans}_f(u,v)}{\lexpans{r}(f)} -1\right|^q\right)^{1/q}.\]
\end{definition}

In \cite{BFN19} the authors rigorously analyzed dimensionality reduction for the above distortion measures. 
The central question they studied is: {\;\it What dimensionality reduction method is optimal for these quality measures and what are the optimal bounds achievable ? In particular, is the Johnson-Lindenstrauss (JL) transform also optimal for the average quality criteria?}

Their analysis of the Gaussian implementation of the JL embedding \cite{IM98} shows that any Euclidean subset can be embedded with $1+\epsilon$ {\it average} distortion ($\ldist{1}$) into $k=O(1/\epsilon^2)$ dimensions. And for more general case of the $q$-moment of distortion, the dimension is $k=O(\max\{1/\epsilon^2,q/\epsilon\})$. However, tight lower bounds were proved only for large values of $q$, 
leaving the question of optimality of the most important case of small $q$, and particularly the most basic  case of $q=1$, unresolved. 

For the additive average measures (stress, energy and others) they prove that a bound of $\epsilon$ can be achieved in dimension $k=O(q/\epsilon^2)$. They showed a tight lower bound on the required dimension only for $q\geq 2$, leaving the basic case of $q=1$ also unresolved. 

In this paper, we prove that indeed the Johnson-Lindenstrauss bounds are best possible for any dimensionality reduction for the full range of $q\geq 1$, for all the average distortion measures defined in this paper. We believe that besides theoretical contribution this statement may have important implications for practical considerations. In particular, it may affect the way the JL method is viewed and used in practice, and the bounds we give may serve a basis for comparison for other methods and heuristics.



\paragraph*{Our results} We show that the guarantees given by the Gaussian random projection dimensionality reduction are the best possible for the average distortion measures. In particular, we prove 

\begin{theorem}\label{thm:average}
Given any integer $n$ and $\Omega(\frac{1}{\sqrt{n}}) < \epsilon < 1$, there exists a $\Theta(n)$-point subset of Euclidean space such that any embedding $f$ of it into $\ell_2^k$ with $\ldist{1}(f)\leq 1+\epsilon$ requires $k=\Omega(1/\epsilon^2)$. 
\end{theorem}

For the more general case of large values of $q$, we show

\begin{theorem}\label{thm:lqdist}
Given any integer $n$, and $\Omega(\frac{1}{\sqrt{n}}) < \epsilon < 1$, and $1 \leq q <= O(\log(\epsilon^2 n)/\epsilon)$, there exists a $\Theta(n)$-point subset of Euclidean space such that any embedding of it into $\ell_2^k$ with $\ell_q$-distortion at most $1+\epsilon$ requires $k=\Omega(q/\epsilon)$. 
\end{theorem}

As $\ell_q$-distortion is monotonically increasing as a function of $q$, the theorems imply the lower bound of $k=\Omega \left(\max\left\{ 1/\epsilon^2, q/\epsilon\right\}\right)$. 

For the additive distortion measures we prove the following theorem:
\begin{theorem}\label{thm:additive}
Given any integer $n$ and $\Omega(\frac{1}{\sqrt{n}}) < \epsilon < 1$, there exists a $\Theta(n)$-point subset of Euclidean space such that any embedding of it into $\ell_2^k$ with any of $Energy_1$, $Stress_1$, $Stress^{*}_1$, $REM_1$ or $\sigma$-distortion bounded above by $\epsilon$ requires $k=\Omega(1/\epsilon^2)$. 
\end{theorem}

Our main proof is of the lower bound for $Energy_1$ measure, which we show to imply the bound in Theorem~\ref{thm:average} and for all measures in Theorem~\ref{thm:additive}, with some small modifications for the stress measures. Furthermore, since all additive measures are monotonically increasing with $q$ the bounds hold for all $q \geq 1$. Therefore Theorems~\ref{thm:average} and \ref{thm:lqdist} together provide a tight bound of $\Omega(\max\{1/\epsilon^2,q/\epsilon\})$ for the $\ell_q$-distortion. Additionally combined with the lower bounds of \cite{BFN19} for $q \geq 2$, Theorem~\ref{thm:additive} provides a tight bound of $\Omega(q/\epsilon^2)$ for all additive measures.

\paragraph*{Techniques}
The proofs of the lower bounds in all the theorems are based on counting argument, as in the lower bound for the worst case distortion proven by \cite{LN17}. We extend the framework of $\cite{LN17}$ to the entire range of $q$ moments of distortion, including the average distortion. As in the original proof we show that there exists a large family $\mathcal{P}$ of metric spaces that are quite different from each other so that if one can embed all of these into a Euclidean space with a small average distortion the resulting image spaces are different too. This implies that if the target dimension $k$ is too 
small there is not enough space to accommodate all the different metric spaces from the family.

Let us first describe the framework of \cite{LN17}.\footnote{The description is based on combining the methods of \cite{LN17,AK17}, and can be also viewed as our $q$-moments bound with $q=\Theta(\log(\epsilon^2 n)/\epsilon)$.} The main idea is to construct a large family of $n$-point subspaces ${\rm I} \subseteq \ell_2^{\Theta(n)}$ so that each subspace in the family can be uniquely encoded using a small number of bits, assuming that each ${\rm I}$ can be embedded with $1+\epsilon$ worst-case distortion in $\ell_2^k$. The sets they construct are such that the information on the inner products between all the points in ${\rm I}$, even if distorted by an additive error of $O(\epsilon)$, enables full reconstruction of the points in the set. In particular, each ${\rm I}$ consists of a zero vector together with the standard basis vectors ${\rm E}$ and an additional set of vectors denoted by ${\rm Y}$. The set ${\rm Y}$ is defined in such a way that $\inp{y}{e}\in \{0, c\epsilon\}$, for a constant $c>1$, for all $(y,e) \in {\rm Y} \times {\rm E}$. The authors then show that a $1+\epsilon$ distortion embedding $f$ of ${\rm I}$ must map all the points into the ball of radius $2$ while preserving all the inner products up to an additive error $\Theta(\epsilon)$, which enables to recover the vectors in ${\rm Y}$. 
The next step is to show that all image points can be encoded using a small number of bits, while preserving the inner product information up to an $\Theta(\epsilon)$ additive error. This can be achieved by carefully discretizing the ball, and applying a map $\tilde{f}$ mapping every point to its discrete image approximation so that $\inp{f(v)}{f(u)}=\inp{{\tilde f(v)}}{{\tilde f(u)}} \pm \Theta(\epsilon)$. For this purpose one may use the method of $\cite{AK17}$ who showed\footnote{The original proof of \cite{LN17} uses a different elegant discretization argument.} that randomly rounding the image points to the points in a small enough grid will preserve all the pairwise inner products within $\Theta(\epsilon)$ additive error with constant probability, and this in turn allows to derive a short binary encoding for each input point. This implies the lower bound on $k=\Omega(\log(\epsilon^2 n)/\epsilon^2)$, for $\epsilon = \Omega(1/\sqrt{n})$.

Let us now explain the challenges in applying this method to the case of bounded average distortion and $q$-moments. Assuming $f:{\rm I} \to \ell_2^k$ has $1+\epsilon$ average distortion neither implies that all images are in a ball of constant radius nor that $f$ preserves all pairwise inner products. The bounded average distortion also does not guarantee the existence of a large subset of ${\rm I}$ with the properties above.  We suggest the following approach to overcoming these issues. First, we add to ${\rm I}$ a large number of ''copies'' of $0$ vectors which enables to argue that a large subset ${\rm \hat{I}} \subseteq {\rm I}$  will be mapped into a constant radius ball, such that the average additive distortion is $\Theta(\epsilon)$. The next difficulty is that if the images would be rounded to the points in a grid using a mapping which would preserve {\it all} pairwise inner products  with $\Theta(\epsilon)$ additive error, then the resulting grid would be too large, which would't allow a sufficiently short encoding. We therefore round the images to a grid with $\Theta(\epsilon)$ additive approximation to the {\it average} of the inner products of ${\rm \hat{I}}$ and thus reduce the size of the grid (and the encoding). 
The next step is showing that the above guarantees imply the existence of a large enough subset of pairs ${\cal Z} \subseteq {\binom{{\rm I}} {2}}$ of special structure, which allows to encode the {\it entire} set ${\rm I}$ with a few bits even if only the partial information about the inner products within ${\cal Z}$ is approximately preserved. In particular, we show that there is a large subset ${\cal Y}^G \subseteq Y$ such that for each point $y \in {\cal Y}^G$ there is a large enough subset ${\cal E}_y \subseteq E$ such that all pairwise inner products $\inp{y}{e}$, where $y \in {\cal Y}^G$ and $e \in {\cal E}_y$, are additively preserved up to $\Theta(\epsilon)$ in the grid embedding, and therefore all the discretized images of these points have short binary encoding. The last step is to argue that this subset is sufficiently large so the knowledge of its approximate inner products possesses enough information in order to recover the entire point set ${\rm I}$ from our small size encoding. As this set still covers only a constant fraction of the pairs, and encoding the rest of the points is far more costly, this forces the dimension and number of points in our instance to be set to $d=\Theta(n)=\Theta(1/\epsilon^2)$, implying a lower bound of $k=\Omega(1/\epsilon^2)$. Finally, we prove that this can extend to arbitrary large subspaces via metric composition techniques. 
To extend these ideas to the general case of $q$-moments of distortion we prove that similar additive approximation distortion bounds hold with high probability of at least $1-e^{-\Theta(\epsilon q)}$. This means that a smaller fraction of the pairs require a more costly encoding, and allows us to set $d = \Theta(n) = \Theta(1/\epsilon^2) \cdot e^{\Theta(\epsilon q)}$, implying a lower bound of $k=\Omega(q/\epsilon)$.

\paragraph*{Related work} The study of ''beyond the worst-case'' distortion analysis of metric embedding was initiated in \cite{KSW09} by introducing partial and scaling distortions. This generated a rich line of follow up work, \cite{ABCDG05, ABN06, AbrahamBN07} just to name a few. The notions of average distortion and $\ell_q$-distortion were introduced in \cite{ABN06} who gave bounds on embedding general metrics in normed spaces. Other notions of refined distortion analysis considered in the literature include such notions as Ramsey type embeddings \cite{BLMN05}, local distortion embeddings \cite{ABN07}, terminal and prioritized distortion \cite{EFN15, EFN15pr}, and recent works on distortion of the $q$-moments\footnote{The notion in these papers, also studied \cite{ABN06,BFN19}, computes the ratio between the average of (or $q$-moments) of new distances to that of original distances, in contrast to the average distortion (or $q$-moments of distortion) measure in Definition \ref{def:av_dist}, which measures the average (or $q$-moments) of pairwise distortions. } \cite{Naor14, Naor21, KushNT21}. 

In applied community, various notions of average distortion are frequently used to infer the quality of heuristic methods \cite{Hei88a, Gro95, ShT04, CDKLM04, SharmaXBL06,Vera07,CS13}. 

However, the only work rigorously analyzing these notions we are aware of is that of \cite{BFN19}. 
They proved lower bounds of $k=\Omega(1/\epsilon)$ for the all additive measures average ($1$-norm) version, and for the average distortion measure ($\ell_1$-distortion), which we improve here to the tight $\Omega(1/\epsilon^2)$ bound. For $q \geq 2$ they gave tight bounds of $\Omega(q/\epsilon^2)$ for all additive measures. For $\lqdist$ they have shown that for $q = \Omega(\log(1/\epsilon)/\epsilon)$ the tight bound of $k=\Omega(q/\epsilon)$ follows from the black-box reduction to the lower bound on the worst case distortion.


\section{Lower bound for average distortion and additive measures}\label{sec:lower_bound}
In this section we prove Theorems \ref{thm:average} and Theorem~\ref{thm:additive}. Using Claim \ref{BFN-relations}, we may focus on proving the lower bound for $\energy_1(f)$ in order to obtain similar lower bounds for $REM_1(f)$ and $\ldist{1}(f)$. In Appendix \ref{app:additive_meas} we show how to change this proof in order to obtain lower bound on $Stress_1(f)$, and further show that the lower bounds for all the additive measures follow from the lower bounds on ${\rm Energy}$ and ${\rm Stress}$. 

We present here the proof of an existence of a bad metric space of size $\hat{n}=\Theta(1/\epsilon^2)$ and show in Appendix \ref{app:arbit_size} how to extend it for metric spaces of arbitrary size ${n} \geq \hat{n}$.

We construct a large family $\mathcal{P}$ of metric spaces, such that each ${\rm I} \in \mathcal{P}$ can be completely recovered by computing the inner products between the points in ${\rm I}$.  
For a given $\epsilon <0$, let $l=\lceil\frac{1}{\gamma^2 \epsilon^2}\rceil$, for some large constant $\gamma > 1$ to be determined later. We will prove $k \geq \frac{c}{\gamma^2}\frac{1}{\epsilon^2}$, for $c<1$, and so we may assume w.l.o.g. that $\epsilon \leq 1/\gamma$, otherwise the statement trivially holds. We will construct point sets ${\rm I} \subset \ell_2^d$, where $d=2l$, each ${\rm I}$ of size $3d=6l=\Theta(1/\epsilon^2)$.

Define a set $O=\{o_j\}_{j=1}^{d}$ of $d$ arbitrary near zero vectors in $\ell_2^d$, i.e., a set of $d$ different vectors such that for all $o_j \in O$, $\norm{o_j}_2\leq \epsilon/100$. Let ${E}=\{e_1, e_2, \ldots, e_{d}\}$ denote the vectors of the standard basis of $\mathbb{R}^{d}$. For a set $S$ of $l$ indices from $\{1, 2, \ldots, d\}$, we define $y_S=\frac{1}{\sqrt{l}}\sum_{j \in S}e_j$. For a sequence of $d$ index sets (possibly with repetitions) $S_1, S_2, \ldots, S_{d}$, let ${\rm Y}[S_1, \ldots, S_{d}]=\{y_{S_1}, \ldots, y_{S_{d}}\}$. Each point set ${\rm I}[S_1, \ldots, S_{d}] \in \mathcal{P}$ is defined as the union of the sets defined above\footnote{We will omit $[S_1, \ldots S_{d}]$ from notation for a fixed choice of the sets.}, i.e., ${\rm I}[S_1, \ldots, S_{d}]=O \cup {E}\cup {\rm Y}[S_1, \ldots, S_{d}]$. The size of the family is $\abs{\mathcal{P}}= {\binom{d}{l}}^d$.
Note that each ${\rm I} \in \mathcal{P}$ is contained in $B_2(1)$, the unit ball of $\ell_2^d$, and has diameter $diam({\rm I}) = \sqrt{2}$. Additionally, for all $e_j \in E$ and $y_S\in Y$ the value of the inner product $\inp{e_j}{y_S}$ determines whether $e_j \in {\rm span}\{e_i | i \in S\}$. In particular, if $\inp{e_j}{y_S}=0$ then $j \not\in S$, and if $\inp{e_j}{y_S}=1/\sqrt{l}\geq (1/2)\gamma\epsilon$ then $j \in S$. 

Assume that for each ${\rm I} \in \mathcal{P}$ there is an embedding $f: {\rm I} \to \ell_2^k$, with $\energy_1(f)\leq \epsilon$. We prove that this implies that $k=\Omega(1/\epsilon^2)$. The strategy is to produce a unique binary encoding of each ${\rm I}$ in the family based on the embedding $f$. Let ${\rm length(I)}$ denote the length of the encoding for each ${\rm I}$, we will show that ${\rm length(I)} \lesssim l^2 + l\cdot k \log(\frac{1}{\epsilon k})$. Since the encoding defines an injective map from $\mathcal{P}$ to $\{0,1\}^{\rm length(I)}$, the number of different sets that can be recovered by decoding is at most $2^{\rm length(I)}$. Now, because $\abs{\mathcal{P}} = \binom{d}{l}^d \geq 2^{2l^2}$ we get that $k \log(\frac{1}{\epsilon k}) \gtrsim l$ and show that this implies the bound on $k\geq\Omega(l)$. 

We are now set to describe the encoding for each ${\rm I}$ and to bound its length.
First, in the following lemma, we show that there exists a large subset ${\rm \hat{I}} \subseteq {\rm I}$ that is mapped by $f$ into a ball of a constant radius in $k$-dimensional space and that the average of the errors in the inner products incurred by $f$ on the subset ${\rm \hat{I}}$ is bounded by $\Theta(\epsilon)$.

\begin{lemma}\label{lemma:average}
For any ${\rm I} \in \mathcal{P}$ let $f:{\rm I} \to \ell_2^k$ be an embedding with $\energy_1(f) \leq \epsilon$, with $\epsilon \leq 1/36$. Let $0<\alpha \leq 1/16$ be a parameter. There is a subset $\hat{{\rm I}} \subseteq {\rm I}$ of size $\abs[\big]{\hat{{\rm I}}}\geq (1-\alpha)\abs{{\rm I}}$ such that $f(\hat{{\rm I}}) \subset B_2\left( 1+\frac{3.01\epsilon}{\alpha} \right)$, and   
$\frac{1}{|{\binom{\hat{{\rm I}}}{2}}|} \sum_{(u,v) \in {\binom{\hat{{\rm I}}}{2}}} \abs[\big]{\inp{f(u)}{f(v)}-\inp{u}{v}} \leq (10+\frac{1}{2\alpha}) \epsilon$.  
\end{lemma}

\begin{proof}
By assumption we have that the following condition holds:
\begin{equation}\label{eq:energy-bound} 
{\energy}_1(f) = \frac{1}{\abs[\big]{\binom{I}{2}}} \sum_{(u,v) \in {\binom{I}{2}}} {\abs[\big]{{\expans}_{f}(u,v)-1}} \leq \epsilon. \end{equation}

This bound implies that
\begin{eqnarray*}\label{eq:O-I} 
\frac{1}{\abs{O}(\abs{I}-1)} \sum_{o_j \in O}\sum_{v \in I, v\neq o_j} {\abs[\big]{{\expans}_{f}(o_j,v)-1}} &\leq &  \frac{1}{\abs{O}(\abs{I}-1)} \sum_{u \neq v \in I} {\abs[\big]{{\expans}_{f}(u,v)-1}}\\
& \leq & \frac{3d(3d-1)}{d(3d-1)} \epsilon = 3\epsilon.
\end{eqnarray*} 

From which follows that 
\begin{equation}\label{eq:o-hat} 
\min_{o_j \in O} \frac{1}{|I|-1} \sum_{v \in I, v\neq o_j} {\abs[\big]{{\expans}_{f}(o_j,v)-1}} \leq 3\epsilon.
\end{equation}
Let $\hat{o} \in O$ denote the point at which the minimum is obtained. We assume without loss of generality that $f(\hat{o})=0$. Let $\hat{I}$ be the set of all $v \in I$ such that $ {\abs[\big]{{\expans}_{f}(\hat{o},v)-1}} \leq \frac{3\epsilon}{\alpha}$. By Markov's inequality, $\abs{\hat{I}} \geq (1-\alpha)\abs[\big]{I}$. 
We have that for all $v\in {\hat{\rm I}}$, $\abs{{\expans}_f(v,\hat{o})-1}=\abs{\frac{\norm{f(v)}_2}{\norm{v-\hat{o}}_2}-1}\leq \frac{3\epsilon}{\alpha}$, and using $\norm{v- \hat{o}}_2 \leq \norm{v}_2 + \norm{\hat{o}}_2 \leq 1+\epsilon/100$, so that
$\norm{f(v)}_2 \leq (1+\frac{3\epsilon}{\alpha})(1+\epsilon/100) \leq 1+ \frac{3.002\epsilon}{\alpha}$, implying that $f(v) \in B_2\left( 1+\frac{3.01\epsilon}{\alpha}\right)$.

For all $(u,v)\in {\binom{\hat{I}}{2}}$ we have:
\begin{eqnarray*}
\abs[\big]{\inp{f(u)}{f(v)}-\inp{u}{v}} & \leq & \frac{1}{2} \left[ \abs[\big]{\norm{f(u)}_2^2-\norm{u}_2^2} + \abs[\big]{\norm{f(v)}_2^2-\norm{v}_2^2}\right] \\
& + &  \frac{1}{2} \left[ \abs[\big]{\norm{f(u)-f(v)}_2^2-\norm{u-v}_2^2} \right] .
\end{eqnarray*}

We can bound each term as follows:
\begin{eqnarray*}
\lefteqn{ \abs[\big]{\norm{f(u)}_2^2-\norm{u}_2^2} = } \\
& = & \abs{\norm{f(u)-f(\hat{o})}_2^2-\norm{u-\hat{o}}_2^2 + \norm{u-\hat{o}}_2^2-\norm{u}_2^2} \\
& \leq & \abs{\norm{f(u)-f(\hat{o})}_2-\norm{u-\hat{o}}_2} \cdot \left(\norm{f(u)-f(\hat{o})}_2+\norm{u-\hat{o}}_2\right) \\
&+& \abs{\norm{u-\hat{o}}_2-\norm{u}_2} \cdot (\norm{u-\hat{o}}_2+\norm{u}_2) \\
& \leq & \norm{u-\hat{o}}_2 \cdot \abs{{\expans}_f(u,\hat{o})-1}\cdot(\norm{f(u)}_2+\norm{u-\hat{o}}_2) + \norm{\hat{o}}_2 \cdot(\norm{u-\hat{o}}_2+\norm{u}_2) \\
& \leq & \left(1+\frac{\epsilon}{100}\right)\abs{{\expans}_f(u,\hat{o})-1}\left(1+\frac{3.002\epsilon}{\alpha}+1+\frac{\epsilon}{100}\right) + \frac{\epsilon}{100} \cdot \left(2+\frac{\epsilon}{100}\right) \\
& \leq & \left(2+\frac{3.01\epsilon}{\alpha}\right)\abs{{\expans}_f(u,\hat{o})-1} + \frac{\epsilon}{40} 
\leq \left(2+\frac{1}{9\alpha}\right)\abs{{\expans}_f(u,\hat{o})-1} + \frac{\epsilon}{40} ,
\end{eqnarray*} 
where we have used $\norm{\hat{o}} \leq \epsilon/100$, $\norm{u-\hat{o}}_2 \leq \norm{u}_2+\norm{\hat{o}}_2 \leq 1+\epsilon/100$, and the bound on the norms of the embedding within $\hat{I}$. 
Additionally, we have that 
\begin{eqnarray*}
\lefteqn{ \abs[\big]{\norm{f(u)-f(v)}_2^2-\norm{u-v}_2^2} = } \\
& = & \abs{\norm{f(u)-f(v)}_2-\norm{u-v}_2}(\norm{f(u)-f(v)}_2+\norm{u-v}_2) \\
& \leq & \norm{u-v}_2 \abs{{\expans}_f(u,v)-1}(\norm{f(u)}_2+\norm{f(v)}_2 + \norm{u-v}_2) \\ 
& \leq & \sqrt{2} \left(2\left(1+\frac{3.002\epsilon}{\alpha}\right) + \sqrt{2}\right) \abs{{\expans}_f(u,v)-1} 
\leq \left(5+\frac{1}{4\alpha}\right) \abs{{\expans}_f(u,v)-1} ,
\end{eqnarray*}
where the second to last inequality holds since $\norm{u-v}_2 \leq diam(I) = \sqrt{2}$.
It follows that:
\begin{eqnarray}\label{eq:sum-innerprod}
\lefteqn{ \frac{1}{|{\binom{\hat{I}}{2}}|} \sum_{(u,v) \in {\binom{\hat{I}}{2}}} \abs[\big]{\inp{f(u)}{f(v)}-\inp{u}{v}} \leq } \\ \notag
 & \leq & \left(2+\frac{1}{9\alpha}\right)\cdot \frac{1}{|{\binom{\hat{I}}{2}}|} \left(\frac{|\hat{I}|-1}{2}\right) \sum_{u \in \hat{I}, u \neq \hat{o}} \abs{{\expans}_f(u,\hat{o})-1} \\  &+& \frac{1}{2}\left(5+\frac{1}{4\alpha}\right)\cdot \frac{1}{|{\binom{\hat{I}}{2}}|} \sum_{(u,v) \in {\binom{\hat{I}}{2}}} \abs{{\expans}_f(u,v)-1} + \frac{\epsilon}{40}. \nonumber
\end{eqnarray}
By definition of $\hat{I}$, and using~(\ref{eq:o-hat}) we have that
\begin{eqnarray*}
\frac{1}{|{\binom{\hat{I}}{2}}|}\left(\frac{|\hat{I}|-1}{2}\right) \sum_{u \in \hat{I}, u \neq \hat{o}} \abs{{\expans}_f(u,\hat{o})-1} &=&  \frac{1}{|\hat{I}|}\sum_{u \in \hat{I}, u \neq \hat{o}} \abs{{\expans}_f(u,\hat{o})-1} \\
 &\leq& \frac{1}{|I|} \sum_{u \in I, u \neq \hat{o}} \abs{{\expans}_f(u,\hat{o})-1} \leq 3\epsilon.
\end{eqnarray*}

Therefore \eqref{eq:sum-innerprod} yields that
\begin{eqnarray*}
\lefteqn{\frac{1}{|{\binom{\hat{I}}{2}}|} \sum_{(u,v) \in {\binom{\hat{I}}{2}}} \abs[\big]{\inp{f(u)}{f(v)}-\inp{u}{v}} \leq}\\ 
&\leq& \left(2+\frac{1}{9\alpha}\right)\cdot 3\epsilon + \frac{1}{2}\left(5+\frac{1}{4\alpha}\right) \cdot \frac{1}{|{\binom{\hat{I}}{2}}|} \sum_{(u,v) \in {\binom{\hat{I}}{2}}} \abs{{\expans}_f(u,v)-1} + \frac{\epsilon}{40} \\
& \leq & \frac{1}{2}\left(5+\frac{1}{4\alpha}\right) \cdot \frac{1}{|{\binom{\hat{I}}{2}}|} \sum_{(u,v) \in {\binom{\hat{I}}{2}}} \abs{{\expans}_f(u,v)-1} + \left(7+ \frac{1}{3\alpha}\right) \epsilon.
\end{eqnarray*}

Now, we have that
\begin{eqnarray*}
\frac{1}{|{\binom{\hat{I}}{2}}|} \sum_{(u,v) \in {\binom{\hat{I}}{2}}} \abs{({\expans}_f(u,v))-1} \leq 
\frac{6}{5} \frac{1}{|{\binom{I}{2}}|} \sum_{(u,v) \in {\binom{I}{2}}} \abs{({\expans}_f(u,v))-1} \leq \frac{6}{5} \epsilon,
\end{eqnarray*}
using $|\hat{I}| \geq (1-\alpha)|I|$, so that $\alpha \leq 1/16$ we have $|{\binom{\hat{I}}{2}}| \geq (1-\frac{1}{3(1-\alpha)d})(1-\alpha)^2 \cdot |{\binom{I}{2}}| \geq \frac{5}{6} |{\binom{I}{2}}|$ and applying \eqref{eq:energy-bound}. 
Finally, we obtain
\begin{eqnarray*}
\frac{1}{|{\binom{\hat{I}}{2}}|} \sum_{(u,v) \in {\binom{\hat{I}}{2}}} \abs[\big]{\inp{f(u)}{f(v)}-\inp{u}{v}}
\leq \frac{6}{5}\cdot \frac{1}{2}\left(5+\frac{1}{4\alpha}\right) \epsilon + \left(7+ \frac{1}{3\alpha}\right) \epsilon \leq \left(10 +\frac{1}{2\alpha}\right) \epsilon.
\end{eqnarray*}

\end{proof}

We have shown thus far that for the large subset ${\hat{\rm I}}$ of the set ${\rm I}$, the average of the inner products between the images equals up to an additive factor $\Theta(\epsilon)$ to the average of the inner products between the original points. Moreover, all the images of ${\hat{\rm I}}$ are in the constant radius ball. We next show that rounding these images to the (randomly chosen) points of the sufficiently small grid will not change the sum of the inner products too much, implying that instead of encoding the original images ${f(\hat{{\rm I}})}$ we can encode its rounded counterpart. To show this, we use a technique of randomized rounding as proposed in \cite{AK17}.

\begin{lemma}\label{lemma:grid}
Let $X\subset \ell_2^k$ such that $X\subset B_2(r)$. For $\delta < r/\sqrt{k}$ let $G_\delta \subseteq B_2(r)$ denote the intersection of the $\delta$-grid with $B_2(r)$. There is a mapping $g : X \to G_\delta$ such that   
$\frac{1}{|{\binom{X}{2}}|} \sum_{(u,v) \in {\binom{X}{2}}} \abs{\inp{g(u)}{g(v)} - \inp{u}{v}} \leq  3\delta r$, and the points of the grid can be represented using $L_{G_\delta} = k\log(4r/(\delta \sqrt{k}))$ bits.
\end{lemma}

\begin{proof}
For each point $v \in X$ randomly and independently match a point $\tilde{v}=g(v)$ on the grid by rounding each of its coordinates $v_i$ to one of the closets integral multiplies of $\delta$ in such a way that $E[\tilde{v}_i]=v_i$. This distribution is given by assigning $\left\lceil  \frac{v_i}{\delta} \right\rceil \delta$ with probability $p=\left( \frac{v_i}{\delta}-\left\lfloor \frac{v_i}{\delta} \right\rfloor \right)$, and assigning $\left\lfloor  \frac{v_i}{\delta} \right\rfloor \delta $ with probability $1-p$. 
For any $u \neq v \in X$ we have:

\begin{eqnarray*} \E\left[\abs{\inp{\tilde{u}}{\tilde{v}} - \inp{u}{v}}\right] &\leq &  \E\left[ \abs{\inp{\tilde{u}-u}{v}} \right] + \E\left[\abs{\inp{\tilde{u}}{\tilde{v}-v}} \right] \\  &\leq& \left(\E\left[(\inp{\tilde{u}-u}{v})^2 \right]\right)^{1/2} + \left( \E\left[(\inp{\tilde{u}}{\tilde{v}-v})^2 \right] \right)^{1/2}, \end{eqnarray*}
where the last inequality is by Jensen's. We bound each term separately. 

\begin{eqnarray*}
\lefteqn{\E[({\inp{\tilde{u}-u}{v}})^2] = \E\left[ \left({\sum_{i=1}^{k}{(\tilde{u}_i-u_i)v_i}} \right)^2 \right] =} \\
 &=&\sum_{i=1}^{k} v_i^2\; \E\left[({\tilde{u}_i-u_i})^2 \right]  + 2\sum_{1\leq i \neq j \leq k}v_i v_j\E[\tilde{u_i}-u_i]\E[\tilde{u_j}-u_j] \leq \delta^2 \norm{v}_2^2\end{eqnarray*}
since $\abs{\tilde{u}_i-u_i} \leq \delta$ and $E[\tilde{u}_i]=u_i$. Similarly, for the second term we have
\begin{align}\E \left[({\inp{\tilde{u}}{\tilde{v}-v}})^2 \right] & =\E\left[ \left({\sum_{i=1}^{k}{\tilde{u}_i(\tilde{v}_i-v_i)}} \right)^2 \right] \leq \sum_{i=1}^{k}\E\left[{\tilde{u}_i}^2\right]\E\left[ ({\tilde{v}_i-v_i})^2 \right] \\ \notag & + 2\sum_{1\leq i \neq j \leq k}\E[\tilde{u}_i\tilde{u_j}(\tilde{v}_i-v_i)]\E[\tilde{v}_j-v_j] \leq \delta^2\sum_{i=1}^{k}\E[\tilde{u}_i^2], \notag \end{align}
because the random variables $\tilde{u}_i$ and $\tilde{v}_i$ are independent. We also have that
\[ \sum_{i=1}^{k} \E[\tilde{u}_i^2] = \sum_{i=1}^{k} \E[(u_i + (\tilde{u}_i-u_i))^2]  = \sum_{i=1}^{k} \left( u_i^2 + 2 u_i \E[\tilde{u}_i-u_i] + \E[(\tilde{u}_i-u_i)^2] \right) \leq \norm{u}_2^2 + \delta^2 k. \]
Therefore, putting all together, 
$\E\left[\abs{\inp{\tilde{u}}{\tilde{v}} - \inp{u}{v}}\right] \leq \delta r + \delta(r^2 + \delta^2 k)^{1/2} \leq 2\delta r + \delta^2\sqrt{k} \leq 3\delta r$.

The bound on the average difference in inner product in the lemma follows by the linearity of expectation, and the implied existence of a mapping with bound at most the expectation.
The upper bound on the representation of the grid points was essentially given in \cite{AK17}: The $i$th coordinate of a point $x$ on the grid is given by a sign and an absolute value $n_i\delta$, where $0\leq n_i \leq r/\delta$ are integers satisfying $\sum_{1\leq i \leq k}n_i^2 \leq (r/\delta)^2$. So can be represented by the sign and their binary representation of size at most $\sum_{i=1}^k (\log(n_i)+1)$, which is maximized when all $n_i^2$'s are equal, which gives the bound of $k\log(4r/(\delta\sqrt{k}))$.\end{proof}

Combining the lemmas we obtain:
\begin{corollary}\label{corr:combine}
For any $I \in \mathcal{P}$ let $f:I \to \ell_2^k$ be an embedding with $\energy_1(f) \leq \epsilon$, with $\epsilon \leq 1/36$. Let $0<\alpha \leq 1/16$ be a parameter. There is a subset $\hat{I} \subseteq I$ of size $\abs[\big]{\hat{I}}\geq (1-\alpha)\abs{I}$ such that there is a set $G \subset \ell_2^k$ and a map $g: \hat{I} \to G$   
such that
\[ \frac{1}{\abs[\big]{\binom{\hat{I}}{2}}}\sum_{(u,v) \in {\binom{\hat{I}}{2}}} \abs[\big]{\inp{g(f(u))}{g(f(v))}-\inp{u}{v}} \leq \left( 13+\frac{0.76}{\alpha} \right)\epsilon, \]
and the points in $G$ can be uniquely represented by binary strings of length at most $L_G = k\log(4r/(\epsilon \sqrt{k}))$ bits, where $r < 1+0.09\frac{1}{\alpha}$.
\end{corollary}
\begin{proof}
The corollary follows by applying Lemma~\ref{lemma:average} followed by Lemma \ref{lemma:grid} with $X=\hat{I}$ and $\delta=\epsilon$. Note that we may assume that $\epsilon = \delta < 1/\sqrt{k} < r/\sqrt{k}$, as otherwise we are done.
\end{proof}

We are now ready to obtain the main technical consequence which will imply the lower bound. It shows that the very special subset of pairs from ${\rm I}$ has all the inner products preserved up to an additive $\Theta(\epsilon)$ and can be encoded using a few bits per point.
\begin{corollary}\label{corr:main}
For any $I \in \mathcal{P}$ let $f:I \to \ell_2^k$ be an embedding with $\energy_1(f) \leq \epsilon$, with $\epsilon \leq 1/36$. Let $0<\alpha \leq 1/16$ and $0 < \beta$ be parameters.
There is a subset of points $G$ 
that satisfies the following: there is a subset ${\mathcal{Y}^G} \subseteq Y$ of size $\abs[\big]{\mathcal{Y}^G}\geq (1-3\alpha-\frac{3}{\sqrt{2}}\beta)\abs{Y}$ such that for each $y \in \mathcal{Y}^G$ there is a subset $\mathcal{E}^G_{y} \subseteq E$ of size $\abs[\big]{\mathcal{E}^G_{y}} \geq (1-3\alpha-\frac{3}{\sqrt{2}}\beta)\abs[\big]{E}$ such that for all pairs $(y,e) \in \mathcal{Y}^G \times \mathcal{E}^G_y$ we have:
$\abs[\big]{\inp{g(f(y))}{g(f(e))}-\inp{y}{e}} \leq \frac{1}{\beta^2}\left( 13+\frac{0.76}{\alpha} \right)\epsilon$, 
where $g: \mathcal{Y}^G \cup \{\mathcal{E}^G_y\}_{y\in \mathcal{Y^G}} \to G$. Moreover, the points in $G$ can be uniquely represented by binary strings of length at most $L_G = k\log(4r/(\epsilon \sqrt{k}))$ bits, where $r < 1+0.09\frac{1}{\alpha}$.
\end{corollary}
\begin{proof}
 Applying Corollary~\ref{corr:combine} and Markov's inequality there are at most $\beta^2$ fraction of pairs $(u,v) \in {\binom{\hat{{\rm I}}}{2}}$ such that $\abs[\big]{\inp{g(f(u))}{g(f(v))}-\inp{u}{v}} > \frac{1}{\beta^2}\left( 13+\frac{0.76}{\alpha} \right)\epsilon$. It follows that the number of pairs in $Y \times E$ that are in ${\binom{\hat{I}}{2}}$ and have this property is at most $\beta^2\cdot \frac{3d(3d-1)}{2} \leq \frac{9}{2} \beta^2 \cdot d^2$. Therefore there can be at most $ \frac{3}{\sqrt{2}} \beta d$ points in $u \in Y$ such that there are more than $\frac{3}{\sqrt{2}} \beta d$ points in $v \in E$ with this property. Since there are at most $3\alpha d$ points in each of $Y$ and $E$ which may not be in $\hat{I}$ we obtain the stated bounds on the sizes of $\abs{\mathcal{Y}^G}$ and $\abs{\mathcal{E}^G_{y}}$.
\end{proof}

In the next subsection we show that preserving such a special and partial information about the inner products in ${\rm I}$ suffices to uniquely encode the whole instance with a small number of bits. 

\subsection{Encoding algorithm}\label{sec:encoding}

Let $t = 8$. We set $\alpha={1}/({12t})$, $\beta={1}/({\sqrt{2}t})$, which implies that $r\leq 10$. Therefore, by Corollary \ref{corr:main}, we can find a subset $G \subseteq B_2(10)$, and a mapping $g: f(I) \to G$, and a subset $\mathcal{Y}^G \subseteq Y$, with $\abs{\mathcal{Y}^{G}}\geq \left(1-\frac{1}{t} \right)\abs{Y}$, where for all $y \in \mathcal{Y}^G$ we can find a susbet $\mathcal{E}_y^G \subseteq E$ with $\abs{\mathcal{E}_y^{G}} \geq \left(1-\frac{1}{t} \right)\abs{E}$, such that for all pairs $(e, y) \in \mathcal{Y}^G \times \mathcal{E}_y^G$ the inner products $ \abs[\big]{\inp{g(f(y))}{g(f(e))}-\inp{y}{e}}\leq 12000 \epsilon$. Moreover, each point in $G$ can be uniquely encoded using at most $L_G= k\log(40/(\epsilon \sqrt{k}))$ bits.

We first encode all the points $Y\setminus\mathcal{Y}^G$. For each $y_S \in Y\setminus\mathcal{Y}^G$ we explicitly write down a bit for each $e_i \in E$ indicating whether $e_i \in S$. This requires $d$ bits for each $y_S$ and in total at most $\left(\frac{1}{t} \right)d^2$ bits for the subset $Y\setminus\mathcal{Y}^G$. The next step is to encode all the points in $\{\mathcal{E}^G_y\}_{y\in \mathcal{Y}^G}$ in a way tat will enable to recover all the vectors in the set together with the indeces. We can do that by writing an ordered list containing $d$ strings (one for each vector in the set $E$, according to its order). Each string is of length $L_G$ bits, where each point $e_i \in \{\mathcal{E}^G_y\}_{y\in \mathcal{Y}^G}$ is encoded by its representation in $G$, i.e., $g(f(e_i))$, and rest of points (if there are any) are encoded by zeros. This gives an encoding of total length $dL_G$ bits.

Now we can encode the points in $\mathcal{Y}^G$. Each $y_S \in \mathcal{Y}^G$ is encoded by the encoding of $g(f(y_S))$ using $L_G$ bits, and in addition we add the encoding of the set of indices of the points in $E\setminus\mathcal{E}^G_{y_S}$, using at most  $\log{ \binom{d}{(1/t)d}} \leq (1/t)d \log(e t)$ bits. Note that this information is not enough in order to recover the vector $y_S$, as we can't deduce whether $i\in S$ for $e_i \in E\setminus\mathcal{E}^G_{y_S}$. So we add this information explicitly, by writing down whether $i \in S$ for each $e_i \in E\setminus\mathcal{E}^G_{y_S}$, using at most $(1/t)d$ bits. In total, it takes $L_G +(1/t)d \log(e t) + (1/t)d $ bits per point in $\mathcal{Y}^G$. 

Therefore, each instance ${\rm I} \in \mathcal{P}$ can be encoded using at most
\[(1/t)d^2 + dL_G+ \abs{\mathcal{Y}^G}\cdot(L_G+d(1/t)\log(e t)+(1/t)d) \leq (1/t) d^2 (2+\log(e t)) +2dL_G\] bits, since $\abs{\mathcal{Y}^G}\leq d$. For our choice of $t=8$, this is at most $\frac{7}{8}d^2 + 2dL_G$.

\subsection{Decoding algorithm}\label{sec:decoding}
To recover the instance ${\rm I}$ from the encoding it is enough to recover the vectors $Y$, since the set of points $O$ and $E$ is the same in each ${\rm I}$. We first recover the set $Y\setminus \mathcal{Y}^G$ in a straightforward way from its naive encoding. 

To recover a point $y_{S} \in \mathcal{Y}^G$ we need to know for each $e_i\in E$ whether $i \in S$. 
An important implication of Corollary \ref{corr:main} is that given $g(f(e_i))$ and $g(f(y_{S}))$ of any pair $(e_i, y_{S}) \in \mathcal{Y}^G \times \mathcal{E}_{y_{S}}^G$, we can decide whether $i \in S$ by computing $\inp{g(f(e_i))}{g(f(y_{S}))}$. Recall that if $i \not \in S$ then $\inp{e_i}{y_{S}}=0$, and if $i \in S$ then $\inp{e_i}{y_{S}}\geq (1/2)\gamma\epsilon$. Therefore, by setting $\gamma = 48001$ we have that if $\inp{g(f(e_i))}{g(f(y_{S}))} \leq 12000\epsilon$, then $i\not \in S$, and $i\in S$ otherwise. We can recover each $g(f(y_{S}))$ for $y_{S} \in \mathcal{Y}^G$ from its binary representation. Next, we recover the set of indices of the points in $E\setminus \mathcal{E}^G_{y_{S}}$, from which we deduce the set of indices of the points $e_i \in \mathcal{E}^G_{y_{S}}$. This gives the information about the set $\{g(f(e_i))\}_{e_i \in \mathcal{E}^G_{y_{S}}}$. At this stage we have all the necessary information to compute the inner products $\inp{g(f(y_{S}))}{g(f(e_i))}$ for all the pairs $y_{S}$ and $e_i$ that enable us to correctly decide whether $i \in S$. Finally, for the rest points $e\in E\setminus \mathcal{E}^G_{y_{S}}$ we have a naive encoding which explicitly states whether $e$ is a part of $y_{S}$. 

\subsection{Deducing the lower bound} 
From the counting argument, the maximal number of different sets that can be recovered from the encoding of length at most $\rho=\frac{7}{8}d^2 +2dL_G$ is at most $2^\rho$. This implies $\frac{7}{8}d^2 +2dL_G \geq \log \abs{\mathcal{P}}$. On the other hand, the size of the family is $\abs{\mathcal{P}}={\binom{d}{l}}^d$. Recall that we have set $d=2l$ so we have that $\abs{\mathcal{P}} \geq {\binom{2l}{l}}^{2l} \geq \left(2^{(2l-1)}/\sqrt{l}\right)^{2l}\geq 2^{4l^2-2l\log l} \geq 2^{3.9 l^2}$, where the last estimate follows from our assumption on $\epsilon$. 
Therefore,  
$\frac{7}{2}l^2+4lL_G \geq 3.9l^2$, implying $L_G \geq (1/10)l$,
where $L_G = k\log(40/(\epsilon \sqrt{k})) = \frac{1}{2}k\log\left( 16(\frac{10}{\epsilon})^2\frac{1}{k}\right) $. This implies that 
$k \log\left( 16\left(\frac{10}{\epsilon}\right)^2\frac{1}{k}\right) \geq (1/5)l \geq 1/(5 \gamma^2 \cdot \epsilon^2)$.
Setting $x= k\cdot (5 \gamma^2 \cdot \epsilon^2)$ we have that
\[1 \leq x \log\left( \frac{0.5}{x} \cdot 2^{14}\gamma^2 \right) = x\log(0.5/x) + x\log\left(2^{14}\gamma^2 \right) \leq 1/2 + 2x(7+\log \gamma),\]
where the last inequality we have used $x\log(0.5/x) \leq 0.5/(e\ln 2) < 1/2$ for all $x$. This implies the desired lower bound on the dimension:
$k \geq 1/(20 \gamma^2 (7+\log \gamma) \cdot \epsilon^2)$.

\section{Lower bounds for \texorpdfstring{$q$}{q}-moments of distortion}

In this section we prove Theorem \ref{thm:lqdist} which provides a lower bound for $q$-moments of distortion. Similarly, to the proof for $\ell_1$-distortion in Section~\ref{sec:lower_bound}, we prove the theorem first for metric space of fixed size $\hat{n} = O(1/\epsilon^2)\cdot e^{O(\epsilon q)}$, which can be extended for metric spaces of size $\Theta(n)$ for any $n$, by the following variation of lemma proved in \cite{BFN19}:

\begin{lemma}
Let $(X,d_X)$ be a metric space of size $\abs{X}=n>1$, and let $(Y, d_Y)$ be a metric space. Assume that for any embedding $f:X \to Y$ it holds that $\lqdist(f)\geq 1+\epsilon$. For any ${n}\geq \hat{n}$ there is a metric space $Z$ of size $\abs{Z}=\Theta({n})$ such that any embedding $F:Z \to Y$ has $\lqdist(F) \geq 1+\epsilon/2$. Moreover, if $X$ is a Euclidean subset then there is an embedding from $Z$ into a finite dimensional Euclidean space with distortion $1+\delta$ for any $\delta>0$.
\end{lemma}

For simplicity we may assume w.l.o.g that $q \geq \frac{3}{\epsilon}$, otherwise the theorem follows from Theorem~\ref{thm:average} by monotonicity of the $\ell_q$-distortion.
The proof strategy has exactly the same structure as in the proof of Section~\ref{sec:lower_bound}, however the sets $I$ are constructed using different parameters.  
For a given $\epsilon <0$, let $l=\lceil\frac{1}{\gamma^2 \epsilon^2} \rceil$ be an integer for some large constant $\gamma > 1$ to be determined later. We construct point sets ${\rm I} \subset \ell_2^d$, where $d=l \tau$, $\tau = e^{\epsilon q}$, and $|{\rm I}| = 3d$. Assume that for all ${\rm I} \in \mathcal{P}$ there is an embedding $f: I \to \ell_2^k$, with $\lqdist(f)\leq 1+\epsilon$. We show that this implies that $k=\Omega(q/\epsilon)$.

As before the strategy is to produce a unique binary encoding of ${\rm I}$ of length ${\rm length(I)}$. We will obtain that $\abs{\mathcal{P}} = \binom{d}{l}^d \geq (d/l)^{ld}$, which will give that ${\rm length(I)} \geq dl\log(d/l) = dl \log(\tau)$. We will show that this implies the bound on $k\geq\Omega(l \log(\tau)) = \Omega(1/\epsilon^2 \cdot \epsilon q) = \Omega(q/\epsilon)$. 
 

As in the proof of Theorem \ref{thm:average}, we can assume w.l.o.g. that $\epsilon \leq 1/\gamma$, which by the choice of $\gamma$ later on implies $\epsilon<1/36$.

\begin{lemma}\label{lemma:lqdist}
For any $I \in \mathcal{P}$ let $f:I \to \ell_2^k$ be an embedding with $\lqdist(f) \leq 1+\epsilon$, for $\epsilon<1/36$. 
There is a subset $\hat{I} \subseteq I$ of size $\abs[\big]{\hat{I}}\geq (1-3/\tau^4)\abs{I}$ such that $f(\hat{I}) \subset B_2\left( 1+6.02\epsilon \right)$, and for $1-2/\tau^4$ fraction of the pairs $(u,v) \in {\binom{\hat{I}}{2}}$ it holds that
$ \abs[\big]{\inp{f(u)}{f(v)}-\inp{u}{v}} \leq 32\epsilon$. 
\end{lemma}
\begin{proof}
By assumption we have $\left(\lqdist(f)\right)^q = \frac{1}{\abs[\big]{{\binom{I}{2}}}} \sum_{(u,v) \in {\binom{I}{2}}} {\left( {\distort}_{f}(u,v) \right)^q } \leq (1+\epsilon)^q$.

By the Markov inequality there are at least $1-1/\tau^4$ fraction of the pairs $(u,v) \in \abs[\big]{{\binom{I}{2}}}$ such that  $(\distort_{f}(u,v))^q \leq \tau^4 (1+\epsilon)^q \leq (1+\epsilon)^{q}\cdot e^{4\epsilon q}$, implying that $\distort_{f}(u,v) \leq 1+6\epsilon$. Therefore,
\[\abs[\big]{{\expans}_{f}(u,v)-1} \leq \max\{{\expans}_{f}(u,v)-1, 1/{\expans}_{f}(u,v)-1 \} = {\distort}_{f}(u,v)-1 \leq 6\epsilon.\] 

For every $o_j \in O$, let $F_j$ be the set of points $v \in I\setminus\{o_j\}$ such that ${\abs[\big]{{\expans}_{f}(o_j,v)-1}} > 6\epsilon$. Then the total number of pairs $(u,v) \in {\binom{I}{2}}$ with the property that ${\abs[\big]{{\expans}_{f}(u,v)-1}} > 6\epsilon$ is at least $\sum_{j=1}^d |F_j|/2$, implying that there must be a point $\hat{o}=o_{j^*} \in O$ such that $|F_{j*}| \leq \frac{1}{\tau^4}\cdot \frac{3d(3d-1)}{d} \leq \frac{3}{\tau^4} (3d-1)$. Define $\hat{I} = I \setminus F_{j*}$ to be the complement of this set, so that $|\hat{I}| \leq (1-\frac{3}{\tau^4})|I|$.
We assume without loss of generality that $f(\hat{o})=0$. 
Let $\hat{O} = O\cap \hat{I}$.
We have that $\abs{{\expans}_f(v,\hat{o})-1}=\abs{\frac{\norm{f(v)}_2}{\norm{v-\hat{o}}_2}-1}\leq 6\epsilon$, and using $\norm{v- \hat{o}}_2 \leq \norm{v}_2 + \norm{\hat{o}}_2 \leq 1+\epsilon/100$, so that
$\norm{f(v)}_2 \leq (1+6\epsilon)(1+\epsilon/100) \leq 1+ 6.02\epsilon$, implying that $f(v) \in B_2\left( 1+6.02\epsilon \right)$.

Denote by $\hat{G}$ the set of pairs $(u,v) \in {\binom{\hat{I}}{2}}$ satisfying that ${\abs[\big]{{\expans}_{f}(u,v)-1}} \leq 6\epsilon$. To bound the fraction of these pairs from below, we can first bound $|\hat{I}| \geq (1-\frac{3}{\tau^4})|I| \geq \frac{5}{2} d$ and $|\hat{I}|-1 \geq 2d$, using that $\tau > 3$ by our assumption on $q$. Therefore, we have that the fraction of pairs $(u,v) \in {\binom{\hat{I}}{2}}\setminus \hat{G}$ is at most $\frac{1}{\tau^4}\cdot \frac{3d(3d-1)}{|\hat{I}|(|\hat{I}|-1)} \leq \frac{1}{\tau^4} \cdot \frac{9}{5} \leq \frac{2}{\tau^4}.$

Finally, to estimate the absolute difference in inner products over the set of pairs $\hat{G}$ we recall some of the estimates from the proof of Section~\ref{sec:lower_bound}. 
For all $(u,v)\in \hat{G}$ we have:
\begin{eqnarray*}
\abs[\big]{\inp{f(u)}{f(v)}-\inp{u}{v}} & \leq & \frac{1}{2} \left[ \abs[\big]{\norm{f(u)}_2^2-\norm{u}_2^2} + \abs[\big]{\norm{f(v)}_2^2-\norm{v}_2^2} \right] \\
& + & \frac{1}{2}\left[\abs[\big]{\norm{f(u)-f(v)}_2^2-\norm{u-v}_2^2} \right] .
\end{eqnarray*}
We can bound each term as follows:
\begin{eqnarray*}
\lefteqn{ \abs[\big]{\norm{f(u)}_2^2-\norm{u}_2^2} = } \\
& = & \abs{\norm{f(u)-f(\hat{o})}_2^2-\norm{u-\hat{o}}_2^2 + \norm{u-\hat{o}}_2^2-\norm{u}_2^2} \\
& \leq & \abs{\norm{f(u)-f(\hat{o})}_2-\norm{u-\hat{o}}_2}\cdot(\norm{f(u)-f(\hat{o})}_2+\norm{u-\hat{o}}_2)\\
 &+& \abs{\norm{u-\hat{o}}_2-\norm{u}_2}\cdot(\norm{u-\hat{o}}_2+\norm{u}_2) \\
& \leq & \norm{u-\hat{o}}_2\abs{{\expans}_f(u,\hat{o})-1}\cdot(\norm{f(u)}_2+\norm{u-\hat{o}}_2) + \norm{\hat{o}}_2\cdot(\norm{u-\hat{o}}_2+\norm{u}_2) \\
& \leq & \left(1+\frac{\epsilon}{100}\right)\abs{{\expans}_f(u,\hat{o})-1}\left(1+6.02\epsilon+1+\frac{\epsilon}{100}\right) + \frac{\epsilon}{100} \cdot \left(2+\frac{\epsilon}{100}\right) \\
& \leq & \left(2+6.06\epsilon\right)\abs{{\expans}_f(u,\hat{o})-1} + \frac{\epsilon}{40}
\leq \left(2+6.06\epsilon\right)\cdot 6\epsilon + \frac{\epsilon}{40} \leq 14\epsilon,
\end{eqnarray*} 
where we have used $\norm{\hat{o}} \leq \epsilon/100$, $\norm{u-\hat{o}}_2 \leq \norm{u}_2+\norm{\hat{o}}_2 \leq 1+\epsilon/100$, the bound on the norms of the embedding within $\hat{I}$, and the property of pairs in $\hat{G}$.
Additionally, we have that 
\begin{eqnarray*}
\lefteqn{ \abs[\big]{\norm{f(u)-f(v)}_2^2-\norm{u-v}_2^2} = } \\
& = & \abs{\norm{f(u)-f(v)}_2-\norm{u-v}_2}\cdot(\norm{f(u)-f(v)}_2+\norm{u-v}_2) \\
& \leq & \norm{u-v}_2 \abs{{\expans}_f(u,v)-1}\cdot(\norm{f(u)}_2+\norm{f(v)}_2 + \norm{u-v}_2) \\ 
& \leq & \sqrt{2} \left(2\left(1+6.02\epsilon\right) + \sqrt{2}\right) \abs{{\expans}_f(u,v)-1} 
\leq 6 \abs{{\expans}_f(u,v)-1} \leq 36\epsilon,
\end{eqnarray*}
since $\norm{u-v}_2 \leq diam(I) = \sqrt{2}$, and the last step follows using the property of pair in $\hat{G}$. We conclude that for all $(u,v) \in \hat{G}$: $\abs{\inp{f(u)}{f(v)}-\inp{u}{v}} 
\leq \frac{1}{2} \left( 2\cdot 14\epsilon + 36\epsilon \right) = 32\epsilon$.
\end{proof}

As before, the goal is to encode the images of the embedding using a sufficiently small number of bits, by rounding them to the points of a grid of the Euclidean ball via the randomized rounding technique of \cite{AK17} as to preserve the inner product gap. The following lemma provides the probability that this procedure fails.

\begin{lemma}\label{lemma:grid-hoeffding}
Let $X\subset \ell_2^k$ such that $X\subset B_2(r)$. For $\delta \leq r/\sqrt{k}$ let $G_\delta \subseteq B_2(r)$ denote the intersection of the $\delta$-grid with $B_2(r)$. There is a mapping $g : X \to G_\delta$ such that for any $\eta \geq 1$, there is a $1-4e^{-\eta^2}$ fraction of the pairs  $(u,v) \in {\binom{X}{2}}$ such that $\abs{\inp{g(u)}{g(v)} - \inp{u}{v}} \leq  3\sqrt{2}\eta \delta r$, and the points of the grid can be represented using $L_{G_\delta} = k\log(4r/(\delta \sqrt{k}))$ bits.
\end{lemma}
\begin{proof}
For each point $v \in X$ randomly and independently match a point $\tilde{v}$ on the grid by rounding each of its coordinates $v_i$ to one of the closest integral multiplies of $\delta$ in such a way that $E[\tilde{v}_i]=v_i$. 
For any $u \neq v \in X$ we have:
$\abs{\inp{\tilde{u}}{\tilde{v}} - \inp{u}{v}} \leq  \abs{\inp{\tilde{u}-u}{v}} + \abs{\inp{\tilde{u}}{\tilde{v}-v}} $.
Now, $\E[\inp{\tilde{u}-u}{v}] = \sum_{i=1}^{k}{\E[\tilde{u}_i-u_i]v_i}=0$ and 
$\E[\inp{\tilde{u}}{\tilde{v}-v}] = \sum_{i=1}^{k}{\E[\tilde{u}_i]\E[\tilde{v}_i-v_i]}= 0$.
Next, we wish to make use of the Hoeffding bound. We therefore bound each of the terms $|(\tilde{u}_i-u_i)v_i| \leq \delta|v_i|$ and the sum $\sum_{i=1}^k \delta^2 v_i^2 = \delta^2 r$, and $|\tilde{u}_i(\tilde{v}_i-v_i)| \leq \delta(|u_i|+\delta)$, so that 
\[ \sum_{i=1}^k \delta^2 (v_i+\delta)^2 = \delta^2 \sum_{i=1}^k (v_i^2 +2\delta v_i + \delta^2) \leq \delta^2 ( r + 2\delta \norm{v_i}_1 + \delta^2 k) \leq \delta^2 (r^2 + 2\delta r \sqrt{k} + \delta^2 k) \leq 4\delta^2 r^2. \]
Applying the Hoeffding bound we get that 
\[\Pr[ \abs{\inp{\tilde{u}-u}{v}} > \sqrt{2}\eta\delta r ] \leq 2e^{-\eta^2}\] and  $\Pr[ \abs{\inp{\tilde{u}}{\tilde{v}-v}} > 2\sqrt{2}\eta\delta r ] \leq 2e^{-\eta^2}$, and therefore 

\[\Pr[ \abs{\inp{\tilde{u}}{\tilde{v}} - \inp{u}{v}} > 3\sqrt{2}\eta\delta r ] \leq 4e^{-\eta^2}.\] This probability also bounds the expected number of pairs with this property so there must exist an embedding to the grid where the bound stated in the lemma holds.
The bound on the representation size is the same as in Lemma~\ref{lemma:grid}.
\end{proof}

Combining the lemmas we obtain:
\begin{corollary}\label{corr:combine-q}
For any $I \in \mathcal{P}$ let $f:I \to \ell_2^k$ be an embedding with $\lqdist(f) \leq 1+\epsilon$, with $\epsilon \leq 1/36$. There is a subset $\hat{I} \subseteq I$ of size $\abs[\big]{\hat{I}}\geq (1-3/\tau^4)\abs{I}$ 
such that for a fraction of at least $1-6/\tau^4$ of the pairs $(u,v) \in {\binom{\hat{I}}{2}}$ it holds that:
$\abs[\big]{\inp{g(f(u))}{g(f(v))}-\inp{u}{v}} \leq 42\epsilon$,
where $g: \hat{I} \to G$. Moreover, the points in $G$ can be uniquely represented by binary strings of length at most $L_G = k\log(5 \sqrt{q/(\epsilon k)})$ bits.
\end{corollary}
\begin{proof}
The corollary follows by applying Lemma~\ref{lemma:lqdist} followed by Lemma~\ref{lemma:grid-hoeffding} with $X=\hat{I}$ with $\delta = 2\sqrt{\epsilon/q}$ and $\eta = \sqrt{\ln(\tau)}$. Note that we may assume that indeed $2\sqrt{\epsilon/q} = \delta < 1/\sqrt{k} < r/\sqrt{k}$, since otherwise we are done.
Therefore, the increase in the absolute difference of the inner products due to the grid embedding is at most: 
$ 3\sqrt{2}\eta\delta r = 6r\sqrt{2\ln(\tau)\epsilon/q} = 6r\sqrt{2(\epsilon q)\epsilon/q} \leq 10\epsilon$. 
The bound on representation of the grid follows from Lemma \ref{lemma:grid-hoeffding}:
$L_G = k\log(4r/(\delta \sqrt{k})) = k\log(4r \sqrt{q/(\epsilon k)}) \leq k\log(5 \sqrt{q/(\epsilon k)})$.
\end{proof}

We are ready to obtain the main technical consequence which will imply the lower bound:
\begin{corollary}\label{corr:main-q}
For any $I \in \mathcal{P}$ let $f:I \to \ell_2^k$ be an embedding with $\lqdist(f) \leq \epsilon$, with $\epsilon \leq 1/36$. There is a subset of points $G$ 
that satisfies the following: there is a subset ${\mathcal{Y}^G} \subseteq Y$ of size $\abs[\big]{\mathcal{Y}^G}\geq (1-6/\tau^2)\abs{Y}$ such that for each $y \in \mathcal{Y}_G$ there is a subset $\mathcal{E}^G_{y} \subseteq E$ of size $\abs[\big]{\mathcal{E}^G_{y}} \geq (1-6/\tau^2)\abs[\big]{E}$ such that for all pairs $(y,e) \in \mathcal{Y}^G \times \mathcal{E}^G_y$ we have: $\abs[\big]{\inp{g(f(y))}{g(f(e))}-\inp{y}{e}} \leq 42\epsilon$, 
where $g: \mathcal{Y}^G \cup \{\mathcal{E}^G_y\}_{y\in \mathcal{Y^G}} \to G$. Moreover, the points in $G$ can be uniquely represented by binary strings of length at most $L_G = k\log(5 \sqrt{q/(\epsilon k)})$ bits.
\end{corollary}
\begin{proof}
Applying Corollary~\ref{corr:combine-q} we have that there are at most $6/\tau^4$ pairs $(u,v) \in {\binom{\hat{I}}{2}}$ such that \linebreak $\abs[\big]{\inp{g(f(u))}{g(f(v))}-\inp{u}{v}} > 42\epsilon$. It follows that the number of pairs in $Y \times E$ that are in ${\binom{\hat{I}}{2}}$ and have this property is at most $\frac{6}{\tau^4} \cdot \frac{3d(3d-1)}{2} \leq \frac{27}{\tau^4} \cdot d^2$. Therefore there can be at most $ \frac{3\sqrt{3}}{\tau^2} \cdot d$ points in $u \in Y$ such that there are more than $\frac{3\sqrt{3}}{\tau^2} d$ points in $v \in E$ with this property. Since there at most $\frac{3}{\tau^4} \cdot d < \frac{0.5}{\tau^2} \cdot d$ points in each of $Y$ and $E$ which may not be in $\hat{I}$ we obtain the stated bounds on the sizes of $\abs{\mathcal{Y}^G}$ and $\abs{\mathcal{E}^G_{y}}$.
\end{proof}

As before we proceed to show that the whole instance ${\rm I}$ can be encoded using sufficiently small number of bits. 

\subsection{Encoding and decoding}
For a set ${\rm I} \in \mathcal{P}$ let $f:{\rm I} \to \ell_2^k$ be an embedding with $\lqdist(f)=1+\epsilon$, where $\Omega\left( \frac{1}{\sqrt{n}}\right)\leq\epsilon<1/36$, and $q = O(\log(\epsilon^2 n)/\epsilon)$. 
Let $t = \tau^2/6$. Therefore, by Corollary \ref{corr:main-q}, we can find a subset $G \subseteq B_2(2)$, and a mapping $g: f(I) \to G$, and a subset $\mathcal{Y}^G \subseteq Y$, with $\abs{\mathcal{Y}^{G}}\geq \left(1-\frac{1}{t} \right)\abs{Y}$, where for all $y \in \mathcal{Y}^G$ we can find a susbet $\mathcal{E}_y^G \subseteq E$ with $\abs{\mathcal{E}_y^{G}} \geq \left(1-\frac{1}{t} \right)\abs{E}$, such that for all pairs $(e, y) \in \mathcal{Y}^G \times \mathcal{E}_y^G$ the inner products $ \abs[\big]{\inp{g(f(y))}{g(f(e))}-\inp{y}{e}}\leq 42 \epsilon$. Moreover, each point in $G$ can be uniquely encoded using at most $L_G = k\log(5 \sqrt{q/(\epsilon k)})$ bits.  

The encoding is done according to the description in Section~\ref{sec:encoding} so we similarly obtain the following bound on the bit length of the encoding: $(1/t) d^2 (2+\log(e t)) +2dL_G$.

The decoding works in the same way as before for an appropriate choice of $\gamma = 169$. 

\subsection{Deducing the lower bound} In this subsection we show that $k=\Omega(q/\epsilon)$, proving the desired lower bound for the case of $n=3d = O(1/\epsilon^2) \cdot e^{O(\epsilon q)}$. 
From the counting argument, the maximal number of different sets that can be recovered from the encoding of length at most $\rho=(1/t) d^2 (2+\log(e t)) + 2dL_G$ is at most $2^\rho$. This implies 
\[(1/t) d^2 (2+\log(e t))  +2dL_G \geq \log \abs{\mathcal{P}}.\] On the other hand, the size of the family is $\abs{\mathcal{P}}={\binom{d}{l}}^d \geq (d/l)^{ld} = \tau^{ld}$, so that $\log(\abs{\mathcal{P}}) = ld \log(\tau)$.
We therefore derive the following inequality 
\[(1/t) d^2 (2+\log(e t)) +2dL_G \geq ld \log(\tau) \Rightarrow L_G \geq (1/4)l \log(\tau),\]
as 
\[(1/t)d (2+\log(e t)) \leq d (2\log(\tau)+4)/\tau^2 \leq d/(2\tau)\log(\tau) = l\log(\tau)/2,\] using that $\log(\tau) > 4$. 

Recall that $L_G = k\log(5 \sqrt{q/(\epsilon k)}) = \frac{1}{2}k\log \left( 25(q/(\epsilon k)) \right) $. 
This implies that 
\[k \log\left( 25\left(\frac{q}{\epsilon k}\right) \right) \geq (1/2)l\log(\tau) \geq 1/(2 \gamma^2 \cdot \epsilon^2) \cdot \epsilon q = 1/(2 \gamma^2) \cdot q/\epsilon .\]
Setting $x= k\cdot (2 \gamma^2 \cdot \epsilon / q)$ we have that
\[1 \leq x \log\left( \frac{0.5}{x} \cdot 100\gamma^2 \right) = x\log(0.5/x) + x\log\left(100\gamma^2 \right) \leq 1/2 + 2x\log(10\gamma),\]
where the last inequality we have used $x\log(0.5/x) \leq 0.5/(e\ln 2) < 1/2$ for all $x$. This implies the desired lower bound on the dimension:
$k \geq 1/(8 \gamma^2 \log(10\gamma)) \cdot q/\epsilon$.

\bibliographystyle{plainurl}
\bibliography{JL_average_LB}

\begin{thebibliography}{10}

\bibitem{ABCDG05}
Ittai Abraham, Yair Bartal, T-H.~Hubert Chan, Kedar~Dhamdhere Dhamdhere, Anupam
  Gupta, Jon Kleinberg, Ofer Neiman, and Aleksandrs Slivkins.
\newblock Metric embeddings with relaxed guarantees.
\newblock In {\em Proceedings of the 46th Annual IEEE Symposium on Foundations
  of Computer Science}, FOCS '05, page 83–100, USA, 2005. IEEE Computer
  Society.
\newblock \href {https://doi.org/10.1109/SFCS.2005.51}
  {\path{doi:10.1109/SFCS.2005.51}}.

\bibitem{AbrahamBN07}
Ittai Abraham, Yair Bartal, and Ofer Neiman.
\newblock Embedding metrics into ultrametrics and graphs into spanning trees
  with constant average distortion.
\newblock In {\em Proceedings of the Eighteenth Annual ACM-SIAM Symposium on
  Discrete Algorithms}, SODA '07, page 502–511, USA, 2007. Society for
  Industrial and Applied Mathematics.

\bibitem{ABN07}
Ittai Abraham, Yair Bartal, and Ofer Neiman.
\newblock Embedding metrics into ultrametrics and graphs into spanning trees
  with constant average distortion.
\newblock In {\em Proceedings of the 18th annual ACM-SIAM symposium on Discrete
  algorithms}, SODA '07, pages 502--511, Philadelphia, PA, USA, 2007. Society
  for Industrial and Applied Mathematics.
\newblock URL: \url{http://portal.acm.org/citation.cfm?id=1283383.1283437}.

\bibitem{ABN06}
Ittai Abraham, Yair Bartal, and Ofer Neiman.
\newblock Advances in metric embedding theory.
\newblock {\em Advances in Mathematics}, 228(6):3026 -- 3126, 2011.
\newblock URL:
  \url{http://www.sciencedirect.com/science/article/pii/S000187081100288X},
  \href {https://doi.org/10.1016/j.aim.2011.08.003}
  {\path{doi:10.1016/j.aim.2011.08.003}}.

\bibitem{Alon09}
Noga Alon.
\newblock Perturbed identity matrices have high rank: Proof and applications.
\newblock {\em Combinatorics, Probability {\&} Computing}, 18(1-2):3--15, 2009.
\newblock URL: \url{http://dx.doi.org/10.1017/S0963548307008917}, \href
  {https://doi.org/10.1017/S0963548307008917}
  {\path{doi:10.1017/S0963548307008917}}.

\bibitem{AK17}
Noga Alon and Bo'az Klartag.
\newblock Optimal compression of approximate inner products and dimension
  reduction.
\newblock In {\em 2017 IEEE 58th Annual Symposium on Foundations of Computer
  Science (FOCS)}, pages 639--650, 2017.
\newblock \href {https://doi.org/10.1109/FOCS.2017.65}
  {\path{doi:10.1109/FOCS.2017.65}}.

\bibitem{TriMap}
Ehsan Amid and Manfred~K. Warmuth.
\newblock Trimap: Large-scale dimensionality reduction using triplets, 2019.
\newblock \href {http://arxiv.org/abs/1910.00204} {\path{arXiv:1910.00204}}.

\bibitem{BFN19}
Yair Bartal, Nova Fandina, and Ofer Neiman.
\newblock Dimensionality reduction: theoretical perspective on practical
  measures.
\newblock In {\em Advances in Neural Information Processing Systems 32: Annual
  Conference on Neural Information Processing Systems 2019, NeurIPS 2019,
  December 8-14, 2019, Vancouver, BC, Canada}, pages 10576--10588, 2019.
\newblock URL:
  \url{https://proceedings.neurips.cc/paper/2019/hash/94f4ede62112b790c91d5e64fdb09cb8-Abstract.html}.

\bibitem{BLMN05}
Yair Bartal, Nathan Linial, Manor Mendel, and Assaf Naor.
\newblock On metric ramsey-type phenomena.
\newblock {\em Annals of Mathematics}, 162(2):643--709, 2005.
\newblock URL: \url{http://www.jstor.org/stable/20159927}.

\bibitem{CS13}
A.~Censi and D.~Scaramuzza.
\newblock Calibration by correlation using metric embedding from nonmetric
  similarities.
\newblock {\em IEEE Transactions on Pattern Analysis and Machine Intelligence},
  35(10):2357--2370, Oct. 2013.
\newblock URL: \url{doi.ieeecomputersociety.org/10.1109/TPAMI.2013.34}, \href
  {https://doi.org/10.1109/TPAMI.2013.34} {\path{doi:10.1109/TPAMI.2013.34}}.

\bibitem{VL18}
Leena Chennuru~Vankadara and Ulrike von Luxburg.
\newblock Measures of distortion for machine learning.
\newblock In S.~Bengio, H.~Wallach, H.~Larochelle, K.~Grauman, N.~Cesa-Bianchi,
  and R.~Garnett, editors, {\em Advances in Neural Information Processing
  Systems 31}, pages 4891--4900. Curran Associates, Inc., 2018.

\bibitem{CV18}
Leena Chennuru~Vankadara and Ulrike von Luxburg.
\newblock Measures of distortion for machine learning.
\newblock In S.~Bengio, H.~Wallach, H.~Larochelle, K.~Grauman, N.~Cesa-Bianchi,
  and R.~Garnett, editors, {\em Advances in Neural Information Processing
  Systems}, volume~31. Curran Associates, Inc., 2018.

\bibitem{CDKLM04}
Russ Cox, Frank Dabek, Frans Kaashoek, Jinyang Li, and Robert Morris.
\newblock Practical, distributed network coordinates.
\newblock {\em SIGCOMM Comput. Commun. Rev.}, 34(1):113--118, January 2004.
\newblock URL: \url{http://doi.acm.org/10.1145/972374.972394}, \href
  {https://doi.org/10.1145/972374.972394} {\path{doi:10.1145/972374.972394}}.

\bibitem{EFN15pr}
Michael Elkin, Arnold Filtser, and Ofer Neiman.
\newblock Prioritized metric structures and embedding.
\newblock In {\em Proceedings of the Forty-Seventh Annual ACM Symposium on
  Theory of Computing}, STOC '15, page 489–498, New York, NY, USA, 2015.
  Association for Computing Machinery.
\newblock \href {https://doi.org/10.1145/2746539.2746623}
  {\path{doi:10.1145/2746539.2746623}}.

\bibitem{EFN15}
Michael Elkin, Arnold Filtser, and Ofer Neiman.
\newblock {Terminal Embeddings}.
\newblock In {\em Approximation, Randomization, and Combinatorial Optimization.
  Algorithms and Techniques (APPROX/RANDOM 2015)}, volume~40 of {\em Leibniz
  International Proceedings in Informatics (LIPIcs)}, pages 242--264, 2015.

\bibitem{Gro95}
Patrick J.~F. Groenen, Rudolf Mathar, and Willem~J. Heiser.
\newblock The majorization approach to multidimensional scaling for minkowski
  distances.
\newblock {\em Journal of Classification}, 12(1):3--19, 1995.

\bibitem{Hei88a}
W.~J Heiser.
\newblock Multidimensional scaling with least absolute residuals.
\newblock In {\em In H. H. Bock (Ed.) Classification and related methods},
  pages 455--462. Amsterdam: NorthHolland, 1988a.

\bibitem{Indyk01}
P.~Indyk.
\newblock Algorithmic applications of low-distortion geometric embeddings.
\newblock In {\em Proceedings of the 42nd IEEE Symposium on Foundations of
  Computer Science}, FOCS '01, page~10, USA, 2001. IEEE Computer Society.

\bibitem{IndykMat04}
Piotr Indyk and Jiri Matousek.
\newblock Low-distortion embeddings of finite metric spaces.
\newblock In {\em in Handbook of Discrete and Computational Geometry}, pages
  177--196. CRC Press, 2004.

\bibitem{IM98}
Piotr Indyk and Rajeev Motwani.
\newblock Approximate nearest neighbors: Towards removing the curse of
  dimensionality.
\newblock In {\em Proceedings of the Thirtieth Annual ACM Symposium on Theory
  of Computing}, STOC '98, pages 604--613, New York, NY, USA, 1998. ACM.
\newblock URL: \url{http://doi.acm.org/10.1145/276698.276876}, \href
  {https://doi.org/10.1145/276698.276876} {\path{doi:10.1145/276698.276876}}.

\bibitem{JL}
William~B. Johnson and Joram Lindenstrauss.
\newblock Extensions of {L}ipschitz mappings into a {H}ilbert space.
\newblock In {\em Conference in modern analysis and probability (New Haven,
  Conn., 1982)}, pages 189--206. American Mathematical Society, Providence, RI,
  1984.

\bibitem{KSW09}
Jon Kleinberg, Aleksandrs Slivkins, and Tom Wexler.
\newblock Triangulation and embedding using small sets of beacons.
\newblock {\em J. ACM}, 56(6):32:1--32:37, September 2009.
\newblock URL: \url{http://doi.acm.org/10.1145/1568318.1568322}, \href
  {https://doi.org/10.1145/1568318.1568322}
  {\path{doi:10.1145/1568318.1568322}}.

\bibitem{KushNT21}
Deepanshu Kush, Aleksandar Nikolov, and Haohua Tang.
\newblock Near neighbor search via efficient average distortion embeddings.
\newblock In {\em 37th International Symposium on Computational Geometry, SoCG
  2021, June 7-11, 2021, Buffalo, NY, {USA} (Virtual Conference)}, pages
  50:1--50:14, 2021.
\newblock \href {https://doi.org/10.4230/LIPIcs.SoCG.2021.50}
  {\path{doi:10.4230/LIPIcs.SoCG.2021.50}}.

\bibitem{LN17}
Kasper~Green Larsen and Jelani Nelson.
\newblock Optimality of the johnson-lindenstrauss lemma.
\newblock In {\em 2017 IEEE 58th Annual Symposium on Foundations of Computer
  Science (FOCS)}, pages 633--638, 2017.
\newblock \href {https://doi.org/10.1109/FOCS.2017.64}
  {\path{doi:10.1109/FOCS.2017.64}}.

\bibitem{linial}
N.~Linial.
\newblock Finite metric spaces- combinatorics, geometry and algorithms.
\newblock In {\em Proceedings of the ICM}, 2002.

\bibitem{Maush90}
Ji\v{r}\'{\i} {Matou\v{s}ek}.
\newblock {Bi-Lipschitz embeddings into low-dimensional Euclidean spaces.}
\newblock {\em {Commentat. Math. Univ. Carol.}}, 31(3):589--600, 1990.

\bibitem{Mat02}
Ji\v{r}\'{\i} {Matou\v{s}ek}.
\newblock {\em Lectures on Discrete Geometry}.
\newblock Springer-Verlag New York, Inc., Secaucus, NJ, USA, 2002.

\bibitem{Umap}
Leland McInnes, John Healy, Nathaniel Saul, and Lukas Großberger.
\newblock Umap: Uniform manifold approximation and projection.
\newblock {\em Journal of Open Source Software}, 3(29):861, 2018.
\newblock \href {https://doi.org/10.21105/joss.00861}
  {\path{doi:10.21105/joss.00861}}.

\bibitem{Naor14}
Assaf Naor.
\newblock Comparison of metric spectral gaps.
\newblock {\em Analysis and Geometry in Metric Spaces}, 2:2:1--52, 2014.

\bibitem{Naor21}
Assaf Naor.
\newblock {An average John theorem}.
\newblock {\em Geometry and Topology}, 25(4):1631 -- 1717, 2021.
\newblock \href {https://doi.org/10.2140/gt.2021.25.1631}
  {\path{doi:10.2140/gt.2021.25.1631}}.

\bibitem{SharmaXBL06}
Puneet Sharma, Zhichen Xu, Sujata Banerjee, and Sung{-}Ju Lee.
\newblock Estimating network proximity and latency.
\newblock {\em Computer Communication Review}, 36(3):39--50, 2006.
\newblock URL: \url{http://doi.acm.org/10.1145/1140086.1140092}, \href
  {https://doi.org/10.1145/1140086.1140092}
  {\path{doi:10.1145/1140086.1140092}}.

\bibitem{ShT04}
Yuval Shavitt and Tomer Tankel.
\newblock Big-bang simulation for embedding network distances in euclidean
  space.
\newblock {\em IEEE/ACM Trans. Netw.}, 12(6):993--1006, December 2004.
\newblock URL: \url{http://dx.doi.org/10.1109/TNET.2004.838597}, \href
  {https://doi.org/10.1109/TNET.2004.838597}
  {\path{doi:10.1109/TNET.2004.838597}}.

\bibitem{tSNE}
Laurens van~der Maaten and Geoffrey Hinton.
\newblock Visualizing data using t-sne.
\newblock {\em Journal of Machine Learning Research}, 9(86):2579--2605, 2008.
\newblock URL: \url{http://jmlr.org/papers/v9/vandermaaten08a.html}.

\bibitem{Vempala}
Santosh~Srinivas Vempala.
\newblock {\em The random projection method}, volume~65 of {\em DIMACS series
  in discrete mathematics and theoretical computer science}.
\newblock Providence, R.I. American Mathematical Society, 2004.
\newblock URL: \url{http://opac.inria.fr/record=b1101689}.

\bibitem{Vera07}
J.~Fernando Vera, Willem~J. Heiser, and Alex Murillo.
\newblock Global optimization in any minkowski metric: A
  permutation-translation simulated annealing algorithm for multidimensional
  scaling.
\newblock {\em J. Classif.}, 24(2):277--301, September 2007.

\bibitem{PacMap}
Yingfan Wang, Haiyang Huang, Cynthia Rudin, and Yaron Shaposhnik.
\newblock Understanding how dimension reduction tools work: An empirical
  approach to deciphering t-sne, umap, trimap, and pacmap for data
  visualization, 2020.
\newblock \href {http://arxiv.org/abs/2012.04456} {\path{arXiv:2012.04456}}.

\end{thebibliography}

\appendix

\section{Metric spaces of an arbitrary size}\label{app:arbit_size}
In order to extend the lower bound for the input metrics of an arbitrary size $n \geq \hat{n}=\Theta(1/\epsilon^2)$, we use the notion of the metric composition proposed in \cite{BLMN05}. Given a metric space $X$, we compose it with another metric space $Y$ of size ${n}/\abs{X}$ by substituting each point in $X$ with a copy of $Y$. The first observation is that in such composition pairs of the points from different copies constitute a constant fraction of all the points in the space. The second observation is that, loosely speaking, the average error over these pairs is, up to a constant, the average of average errors over different ''copies'' of $X$ in the composition.  

The following lemma is a variant of a lemma that appeared in \cite{BFN19}:

\begin{lemma}\label{lem:metric_comp}

Let $(X,d_X)$ be a metric space of size $\abs{X}=\hat{n}>1$, and let $(Y, d_Y)$ be a metric space. Assume that $\alpha>0$ is such that for any embedding $f:X \to Y$ it holds that $\energy_q(f)\geq \alpha$. For any ${n}\geq \hat{n}$ there is a metric space $Z$ of size $\abs{Z}=\Theta({n})$ such that any embedding $F:Z \to Y$ has $\energy_q(F) \geq \alpha/2$. 

Moreover, if $X$ is a Euclidean subset then there is an embedding from $Z$ into a finite dimensional Euclidean space with distortion $1+\delta$ for any $\delta>0$.

\end{lemma}

We prove the lemma here for completeness. 
We start with definition of the composition of metric spaces given in \cite{BLMN05}:

\begin{definition}
Let $(S,d_S)$, $(T, d_T)$ be finite metric spaces. For any $\beta \geq 1/2$, the $\beta$-composition of $S$ with $T$, denoted by $Z=S_{\beta}[T]$, is a metric space of size $|Z|=|S|\cdot |T|$ constructed in the following way. Each point $u \in X$ is substituted with a copy of the metric space $T$, denoted by $T^{(u)}$. Let $u, v \in S$, and $z_i \neq z_j \in Z$, such that $z_i \in T^{(u)}$, and $z_j\in T^{(v)}$, then

\[ d_Z(z_i, z_j)=\begin{cases}
     \frac{1}{\gamma}\frac{1}{\beta} \cdot d_T(z_i, z_j), & u=v \\
        d_S(u, v), & u\neq v \\

   \end{cases}
\]
where $\gamma=\frac{\max_{t \neq t' \in T} \{d_T(t, t') \}}{ \min_{s \neq s' \in s}\{ d_S(s, s')\}}$.
\end{definition}

\begin{proof}
Given any ${n} \geq \hat{n}$ let $m=\lceil \frac{n}{\hat{n}} \rceil$, and let $T$ be any $m$-point metric space. For any $\beta \geq 1/2$ let $Z$ be the $\beta$-metric composition of $X$ with $T$ (note that the choice of $T$ is arbitrary), and let $N = |Z| = \hat{n} m=\Theta({n})$. Let $F:Z \rightarrow Y$ be any embedding, and consider the set of pairs $B \subseteq {\binom{Z}{2}}$, $B=\{ (z_i, z_j)| z_i \in T^{(u)},  z_j \in T^{(v)}, \forall u\neq v \in X\}$. Then, $|B|=m^2 \cdot {\binom{\hat{n}}{2}}$. Let $q\geq 1$, and note that for all $z_i\neq z_j \in Z$ it holds that $\abs{{\expans}_F(z_i, z_j)-1} \geq 0$. Then 

\[
{({\energy}_q(F))}^q \geq \frac{1}{\binom{N}{2}}\sum\limits_{z_i \neq z_j \in B}{\abs{{\\expans}_{F}(z_i,z_j)-1}}^q \geq \frac{1}{2}\cdot\frac{1}{m^2\binom{\hat{n}}{2}}\sum\limits_{z_i \neq z_j \in B}{\abs{{\\expans}_{F}(z_i,z_j)-1}}^q. \]

Let $\mathcal{X}$ denote the family of all possible $n$-point subsets ${\bar X} \subset Z$, where each point of ${\bar X}$ is chosen from exactly one of the copies $T^{(x_1)}, T^{(x_2)}, \ldots T^{(x_{\hat n})}$. Namely, each ${\bar X}$ is a metric space isometric to $X$. Let $F|_{{\bar X}}$ denote the embedding $F$ restricted to the points of ${\bar X}$. The size of the family $|\mathcal{X}|=m^{\hat{n}}$, and it holds that

\begin{align}\frac{1}{m^{\hat{n}}}\sum\limits_{{\bar X} \in \mathcal{X}}{({{\energy}_q}(F|_{{\bar X}})}^q & = \frac{1}{m^{\hat{n}}}\sum\limits_{{\bar X} \in \mathcal{X}}\frac{1}{\binom{\hat{n}}{2}}\sum\limits_{u,v \in {\bar X}} {\abs{{\expans}_{F}(u,v)-1}}^q \notag \\  & =
\frac{1}{m^{\hat{n}}}\frac{1}{\binom{\hat{n}}{2}} \sum\limits_{z_i\neq z_j \in B} m^{\hat{n}-2}\cdot{\abs{{\expans}_{F}(u,v)-1}}^q \notag \\ \notag & = \frac{1}{\binom{\hat{n}}{2}m^2}\cdot \sum\limits_{z_i\neq z_j \in B}{\abs{{\expans}_{F}(u,v)-1}}^q. \notag \end{align}
%
By the assumption it holds that ${{{\energy}_q}(F|_{{\bar X}})}^q \geq \alpha^q$, implying that $\energy_q(F) \geq \alpha$.

Note that the bound on Energy does not depend on the value of $\beta$. It was shown in \cite{BLMN05} (Proposition $2.12$) that if $X$ and $T$ are both Euclidean subsets, then their composition $Z=X_{\beta}[T]$ embeds into a finite dimensional Euclidean space with distortion $(1+\epsilon)$, for $\beta=O\left(1/\epsilon \right)$. This completes the proof of the lemma.
\end{proof}

The lemma implies that in order to obtain a family of metric spaces of any size $\Theta(n)$ it is enough to compose the metric spaces ${\rm }I$ in the family $\mathcal{P}$, of size $\abs{I}=6l=\hat{n}$ with, for example, an equilateral metric space on $\lceil n/6l \rceil$ points.

\section{More additive distortion measures}\label{app:additive_meas}
In this section we prove Theorem \ref{thm:additive} for the additive distortion measures. We will use some of the observations about the basic relations between the measures made in \cite{BFN19}:




\begin{claim}
For an embedding $f:X \to Y$, for any $r\geq 1$, there is an embedding $f':X \to Y$ such that $\sigma_{1,r}(f) = \energy_1(f')$.
\end{claim}

\begin{claim}
\label{claim:stress_star}
For an embedding $f:X \to Y$ there is an embedding $f':X \to Y$ such that ${\rm Stress}_1(f')\leq 4\cdot{\rm Stress^*}_1(f)$.
\end{claim}

Together with Claim \ref{BFN-relations}, these imply that
\begin{corollary}\label{cor:add_lb}

In order to show  a lower bound for $\ell_1$-distortion, ${\rem_1}$ and $\sigma_{1,r}$ it is enough to lower bound ${\energy}_1$. In order to show a lower bound for ${\rm Stress^*}_1$ it is enough to lower bound ${\rm Stress}_1$.

\end{corollary}

We are ready now to prove Theorem \ref{thm:additive}. We restate it here for convenience: 


\begin{theorem}
Given any integer $n$ and $\Omega(\frac{1}{\sqrt{n}}) < \epsilon < 1$, there exists a $\Theta(n)$-point subset of Euclidean space such that any embedding of it into $\ell_2^k$ with any of $\energy_1$, $Stress_1$, $Stress^{*}_1$, $REM_1$ or $\sigma$-distortion bounded above by $\epsilon$ requires $k=\Omega(1/\epsilon^2)$. 
\end{theorem}

\begin{proof}
We already proved the theorem for ${\energy}_1$, therefore by Corollary \ref{cor:add_lb} it remains to prove it for ${\rm Stress}_1$. First, not that for any embedding $f:X \to Y$ it holds that

\[Stress_1(f)=\frac{ \sum_{u \neq v \in X} |\hat d_{uv}-d_{uv}| }{S[X]}=
\frac{1}{S[X]/{\binom{\abs{X}}{2}}} \frac{\sum_{u \neq v \in X}d_{uv}\abs{{\expans}_f(u,v)-1}}{\binom{\abs{X}}{2}},\]
where $S[X]=\sum_{u \neq v \in X}d_{uv}$. We define 

\[\overline{Stress}_1(f):=\frac{\sum_{u \neq v \in X}d_{uv}\abs{{\expans}_f(u,v)-1}}{\binom{\abs{X}}{2}}.\]

Observe that if $X$ is such that $S[X]/\binom{{\abs{X}}{2}} \leq c$ for a constant $c>0$ then $Stress_1(f)=\Omega(\overline{Stress}_1(f))$. Therefore, it is enough to show that there is a metric space (of an arbitrary size $\Theta(n)$) with at most constant average distance on which the lower bound is obtained for 
$\overline{Stress}_1$. We note that the composition Lemma \ref{lem:metric_comp} also works for constructing arbitrary size metrics for $\overline{Stress}$ notion. 

Therefore, we show that there is a metric space ${\rm I}$ of size $\Theta(1/\epsilon^2)$ such that if any embedding $f:I \to \ell_2^k$ has $\overline{Stress}_1(f)\leq \epsilon$ then $k=\Omega(1/\epsilon^2)$. In addition, ${\rm I}$ is such that $S[I]/{\binom{\abs{I}}{2}} \leq \sqrt{2}$. Since a metric space obtained by the composition in Claim \ref{lem:metric_comp} has the same diameter as the base space (${\rm I}$ in our case), this will complete the proof of the lower bound on an arbitrary size example, and its embedding in Euclidean space would increase this bound further by at most an extra $1+\delta$ factor, for arbitrary $\delta>0$.

We argue that a slight variation on the proof in Section \ref{sec:lower_bound} for ${\energy}_1$ works for ${\overline{Stress}_1}$ as well. 
The initial assumption (\ref{eq:energy-bound}) should be changed accordingly to:
\begin{equation}\label{eq:O-stress-bound} 
{\overline{Stress}}_1(f) = \frac{1}{\abs[\big]{\binom{I}{2}}} \sum_{(u,v) \in \binom{I}{2}} \norm{u-v}_2{\abs[\big]{{\expans}_{f}(u,v)-1}} \leq \epsilon. \end{equation}
So then condition (\ref{eq:o-hat}) defining $\hat{o} \in O$ becomes the point minimizing
\begin{equation}\label{eq:O-stress-o-hat} 
\frac{1}{|I|-1} \sum_{v \in I, v\neq \hat{o}} \norm{v-\hat{o}}_2 {\abs[\big]{{\expans}_{f}(\hat{o},v)-1}} \leq 3\epsilon.
\end{equation}
As before, assume without loss of generality that $f(\hat{o})=0$. Let $\hat{I}$ be the set of all $v \in I$ such that $ \norm{v-\hat{o}}_2 {\abs[\big]{{\expans}_{f}(\hat{o},v)-1}} \leq \frac{3\epsilon}{\alpha}$. By Markov's inequality, $\abs{\hat{I}} \geq (1-\alpha)\abs[\big]{I}$. 
We have that for all $v\in {\hat{\rm I}}$, $\norm{v-\hat{o}}_2 \abs{{\expans}_f(v,\hat{o})-1}=\abs{\norm{f(v)}_2 -\norm{v-\hat{o}}_2}\leq \frac{3\epsilon}{\alpha}$, and using $\norm{v- \hat{o}}_2 \leq \norm{v}_2 + \norm{\hat{o}}_2 \leq 1+\epsilon/100$, so that
$\norm{f(v)}_2 \leq 1+\epsilon/100+\frac{3\epsilon}{\alpha} \leq 1+ \frac{3.001\epsilon}{\alpha}$, implying that $f(v) \in B_2\left( 1+\frac{3.01\epsilon}{\alpha}\right)$.

The main change is in the estimations made in Lemma~\ref{lemma:average} using that for all $(u,v)\in \binom{\hat{I}}{2}$:
\begin{eqnarray*}
\abs[\big]{\inp{f(u)}{f(v)}-\inp{u}{v}} & \leq & \frac{1}{2} \left[ \abs[\big]{\norm{f(u)}_2^2-\norm{u}_2^2} + \abs[\big]{\norm{f(v)}_2^2-\norm{v}_2^2}\right] \\
&+& \frac{1}{2}\left[\abs[\big]{\norm{f(u)-f(v)}_2^2-\norm{u-v}_2^2} \right].
\end{eqnarray*}
Using the bounds we have shown for each term we get:
\begin{eqnarray*}
\abs[\big]{\norm{f(u)}_2^2-\norm{u}_2^2} & \leq & \norm{u-\hat{o}}_2\abs{{\expans}_f(u,\hat{o})-1}(\norm{f(u)}_2+\norm{u-\hat{o}}_2) + \norm{\hat{o}}_2(\norm{u-\hat{o}}_2+\norm{u}_2) \\
& \leq & \left(2+\frac{1}{9\alpha}\right)\norm{u-\hat{o}}_2\abs{{\expans}_f(u,\hat{o})-1} + \frac{\epsilon}{40}.
\end{eqnarray*} 
Similarly,
\begin{eqnarray*}
 \abs[\big]{\norm{f(u)-f(v)}_2^2-\norm{u-v}_2^2} &\leq &\norm{u-v}_2 \abs{{\expans}_f(u,v)-1}(\norm{f(u)}_2+\norm{f(v)}_2 + \norm{u-v}_2) \\ 
& \leq & \left(2\left(1+\frac{3.002\epsilon}{\alpha}\right) + \sqrt{2}\right) \norm{u-v}_2 \abs{{\expans}_f(u,v)-1}\\ 
&\leq& \left(5+\frac{1}{4\alpha}\right) \norm{u-v}_2 \abs{{\expans}_f(u,v)-1} ,
\end{eqnarray*}

The rest of the proof carries on exactly the same as in Lemma~\ref{lemma:average} where the terms of the form \linebreak $\abs{{\expans}_f(u,v)-1}$ are replaced by $\norm{u-v}_2\abs{{\expans}_f(u,v)-1}$.
Recall that each ${\rm I} \in \mathcal{P}$ has ${\rm diam(I)} \leq \sqrt{2}$, so that this bound applied to $S[I]/{\binom{\abs{I}}{2}}$ as well, which completes the proof of the theorem. 

\end{proof}

\end{document}